%% file: main.tex
\documentclass[12pt]{article}
\usepackage{amsmath}
\usepackage{amsbsy}
\usepackage{graphicx}
\usepackage{url} 
\usepackage[colorlinks=true,urlcolor=blue,citecolor=blue,linkcolor=blue,bookmarks=true]{hyperref}
\usepackage{graphicx}
\usepackage{color}
\usepackage{setspace} 
\usepackage{mathtools}
\usepackage{amssymb}
\usepackage{caption}
\usepackage{multirow}
\usepackage{natbib}
\usepackage{algorithm}
\usepackage{algpseudocode}
\usepackage{amsthm}
\usepackage{eqnarray}
\usepackage{mathrsfs}
\usepackage{bm}
\usepackage{bibunits}

\usepackage[utf8]{inputenc}
\DeclareUnicodeCharacter{2212}{-}

\newcommand{\blind}{1}

\newtheorem{theorem}{Theorem}[section]
\newtheorem{proposition}[theorem]{Proposition}

\theoremstyle{definition}

\theoremstyle{remark}

\input{math_commands}

\begin{document}
\begin{bibunit}[mychicago]
\def\spacingset#1{\renewcommand{\baselinestretch}%
{#1}\small\normalsize}  

\spacingset{1}

\if1\blind
{
  \title{\bf Multi-Teacher Knowledge Distillation via Teacher-Informed Mixture Priors}
  \author{
  Luyang Fang$^{1}$,
  Yongkai Chen$^{2}$,
  Jiazhang Cai$^{1}$,
  Ping Ma$^{1,*}$,
  Wenxuan Zhong$^{1,*}$\\[0.5em]
  $^{1}$Department of Statistics, University of Georgia\\
  $^{2}$Department of Statistics, Harvard University\\
  $^{*}$Corresponding authors
}
  \date{}
  \maketitle
} \fi


\if0\blind
{
  \bigskip
  \bigskip
  \bigskip
  \begin{center}
     {\LARGE\bf Multi-Teacher Bayesian Knowledge Distillation }
\end{center}
  \medskip
} \fi

\bigskip
\begin{abstract}

Knowledge distillation is a powerful method for model compression, enabling the efficient deployment of complex deep learning models (teachers), including large language models. However, its underlying statistical mechanisms remain unclear, and uncertainty evaluation is often overlooked, especially in real-world scenarios requiring diverse teacher expertise. To address these challenges, we introduce \textit{Multi-Teacher Bayesian Knowledge Distillation} (MT-BKD), where a distilled student model learns from multiple teachers within the Bayesian framework.
Our approach leverages Bayesian inference to capture inherent uncertainty in the distillation process. We introduce a teacher-informed prior, integrating external knowledge from teacher models and task-specific training data, offering better generalization, robustness, and scalability. Additionally, an entropy-based weighting mechanism adaptively adjusts each teacher's influence, allowing the student to combine multiple sources of expertise effectively. MT-BKD enhances the interpretability of the student model's learning process, improves predictive accuracy, and provides uncertainty quantification. We validate MT-BKD on both synthetic and real-world tasks, including protein subcellular location prediction and image classification. Our experiments show improved performance and robust uncertainty quantification, highlighting the strengths of our MT-BKD framework.


\end{abstract}

\noindent%
\vfill

\newpage
\spacingset{1.9} 


\section{Introduction}\label{sec:intro}

The rapid advancements in deep learning have led to the development of increasingly complex models with billions of parameters, achieving state-of-the-art performance across numerous tasks \citep{kondratyuk2023videopoet, dosovitskiy2020image}. However, deploying these large models in real-world applications is often infeasible due to significant computational and storage demands. Knowledge distillation (KD) offers a powerful solution by transferring knowledge from a large, pre-trained \textit{teacher} model to a smaller, more efficient \textit{student} model \citep{zhao2022decoupled,hinton2015distilling,latif2023knowledge}. By training the student to replicate the teacher’s predictions, KD produces lightweight models with competitive performance, suitable for resource-constrained environments \citep{gou2021knowledge}. This approach is particularly prevalent in the development of large language models, where vendors routinely release both full-scale and lightweight versions, with the latter often trained using KD \citep{touvron2024llama3}.

Many real-world scenarios involve learning from multiple teacher models with diverse expertise. For instance, a model trained on structured data may provide insights different from one specialized in unstructured data. Multi-teacher KD extends KD by leveraging this diversity to improve student learning. Recent research addresses this need with data-dependent weighting schemes that balance each teacher’s influence, enabling the student to capture multiple teachers’ strengths \citep{you2017learning, zhao2024mitigating}.

Despite the empirical success of Multi-teacher KD, significant challenges remain. 
First, the statistical insights of KD and their effects on student-model improvement remain unclear. Although some studies \citep{phuong2019towards, menon2021statistical} have explored its impact, they primarily evaluate improvements in predictive performance. A deeper investigation into the statistical and theoretical aspects is crucial for understanding its full potential.
Second, most existing methods prioritize improving prediction accuracy but lack quantification of uncertainty in its prediction.
This limitation reduces their ability to assess the reliability of the student model’s predictions, restricting their applicability in high-stakes domains where interpretability and robustness are crucial.
Moreover, current approaches often treat teachers independently, combining their outputs through weighted averages without considering potential interactions among the teachers. This simplification can lead to suboptimal knowledge transfer, preventing the student from fully leveraging the combined expertise of multiple teachers.

To address these challenges, we propose a Bayesian framework for multi-teacher knowledge distillation, introducing a teacher-informed mixture prior that integrates guidance from multiple teachers into one probabilistically coherent model. By integrating both teacher knowledge and task-specific data, this prior captures diverse external expertise while remaining closely aligned with the task. This integration of knowledge offers better generalization and robustness than purely data-driven priors used in classical Bayesian analysis, especially in the case of limited training data.
The mixture prior also reflects teacher interactions, allowing the data to favor the teachers that best explain the observations naturally. 
Additionally, we employ an entropy-based weighting scheme that dynamically adjusts each teacher’s influence on a per-sample basis, ensuring that the most knowledgeable teacher has the greatest impact.
Beyond improving the distillation process, our Bayesian formulation naturally equips the student model with uncertainty quantification capabilities. By sampling from the student’s posterior via stochastic gradient Langevin dynamics \citep{welling2011bayesian}, we derive robust uncertainty measures and estimate credible intervals. These uncertainty estimates improve interpretability by quantifying the model's confidence in its predictions, which is crucial for reliable decision-making in critical applications.

\noindent\textbf{Our contributions are summarized as follows:}
\begin{itemize}
    \vspace{-0.3cm}
     \item  Uncertainty Quantification: We develop Bayesian inference tools for the student model to provide uncertainty estimates that improve interpretability and enable robust decision-making.
    \vspace{-0.3cm}
    \item  Teacher-informed Prior:  We introduce a Bayesian KD framework that integrates the external knowledge of teacher models and training data to provide the prior,  offering better generalization, robustness, and scalability.     
    \vspace{-0.3cm}
    \item  Unified Multi-teacher Framework: We leverage a mixture distribution to combine the knowledge of multiple teachers, which captures their interactions, and ensures effective knowledge transfer. Our entropy-driven, sample-specific weighting mechanism prioritizes important teachers for each data point, optimizing the student’s learning process.
    \vspace{-0.3cm}
    \item  Validation: We demonstrate the effectiveness of our approach using synthetic and real-world datasets, highlighting its ability to enhance both predictive performance and uncertainty quantification.
    \vspace{-0.3cm}
\end{itemize}

\section{Related Research}\label{sec:related}

Knowledge distillation (KD) has been widely studied for its effectiveness and limitations. The original work by \citet{hinton2015distilling} transfers knowledge by minimizing the KL divergence between teacher and student outputs. Later studies explored strategies like mimicking intermediate-layer representations \citep{zagoruyko2016paying,phuong2019towards,shen2022self}.
To leverage diverse teacher expertise, researchers have proposed multi-teacher KD strategies. 
\citep{you2017learning, zhang2022confidence, liu2020adaptive, chen2022knowledge, gui2023mt4mtl, zhao2024mitigating}. 
For instance, \citet{you2017learning} align intermediate-layer representations across teachers using relative dissimilarity constraints and aggregate their guidance through a voting mechanism. 
\citet{zhao2024mitigating} combine multi-teacher KD with adversarial training by employing specialized teachers for clean and adversarial data, enhancing both accuracy and robustness.
However, existing multi-teacher KD approaches predominantly focus on improving predictive accuracy and often overlook interactions between teachers or the uncertainty in the student model’s predictions.



Efforts to incorporate Bayesian principles into KD have led to methods such as \textit{Bayesian Dark Knowledge} \citep{korattikara2015bayesian}, which distills the teacher model’s posterior predictive distribution into a compact student model, enabling uncertainty quantification. Subsequent extensions \citep{wang2018adversarial, malinin2019ensemble, vadera2020generalized} refine this framework, but these methods require access to the teacher model's posterior predictive distribution, which is often impractical for large-scale pre-trained models.

Formulating data-based priors has a long history in Bayesian statistics \citep{robbins1992empirical}. Some studies incorporate information from previous work to construct priors \citep{ibrahim2015power}. In particular, the power prior is derived by raising the likelihood of historical data to a specified power. However, its properness requires the covariate matrix of either historical or current data to have full column rank \citep{chen2000power}.
Another approach is to use synthetic data to define the prior distribution by training a simpler model and generating samples from its predictive distribution \citep{huang2020catalytic,bernardo1979reference}. In cases with limited expert knowledge, this method is effective. However, with powerful teacher models, MT-BKD can integrate both expert insights and current data to create the prior.

\section{Preliminaries}\label{sec:Preliminaries}
We begin by introducing the preliminaries of deep neural network models, followed by an overview of knowledge distillation and its extension to multi-teacher knowledge distillation. 

\noindent\textbf{\textit{Deep Neural Network Model.}}
For the convenience of presentation, we introduce our method in the context of classification problems. 
In the classification problems, we are given the training sample $\mathcal{D}=\{(\bx_i, \bfy_i)\}_{i=1}^N $, where $\bx_i \in \Omega_\bx \subset \RR^m$ is the predictor and $\bfy_i=(y_{i1},\ldots,y_{iK})^T$ is the corresponding label represented in the form of indicator vector. Specifically, if the $i^{\text {th }}$ data point belongs to class $k$, we have $y_{i k}=1$ and $y_{i j}=0$ for all $j \neq k$, where $i=1, \ldots, N, j=1, \ldots, K$. 
Each data point is independently collected across $i \in \{1,\ldots,N\}$. 
Our goal is to find a function $\bfh: \mathbb{R}^m \to [0,1]^K$ from a function class 
that approximates the conditional probability $\mathbb{P}(y_{ik} = 1 \mid \bx_i)$ for each $k \in \{1, \ldots, K\}$. Specifically, we seek a function $\bfh(\bx_i; \bstheta) = (h_1(\bx_i; \bstheta), \ldots, h_K(\bx_i; \bstheta))^T$, where $h_k(\bx_i; \bstheta)$ corresponds to the approximation of $\mathbb{P}(y_{ik} = 1 \mid \bx_i)$.  The function class can be very general, ranging from logistic regressions \citep{Faraway2016extending}, nonlinear regressions \citep{Bates1988nonlinear},  to deep neural networks \citep{Goodfellow-et-al-2016}.
In this article, we focus on finding $\bfh$ in the function class of deep neural networks, including feed-forward networks, convolutional neural nets, transformers, etc, and please see \cite{Zhang2021dive} for more details.
For example, a feed-forward network
utilizes functions from the following compositional function class:
\begin{equation}\label{dnn_def}
\bfh(\bx ; \bstheta)=\boldsymbol{\sigma}_{L}\left(\mathbf{W}_L \boldsymbol{\sigma}_{L-1}\left(\mathbf{W}_{L-1} \cdots \boldsymbol{\sigma}_2\left(\mathbf{W}_2 \boldsymbol{\sigma}_1\left(\mathbf{W}_1 \bx\right)\right)\right)\right),
\end{equation}
where $\bstheta=\left\{\mathbf{W}_1, \ldots, \mathbf{W}_L\right\}$ are the parameters of the model,  and $\boldsymbol{\sigma}_l$, $l=1, \ldots, L$ are some fixed nonlinear activation functions. 
A popular choice of $\boldsymbol{\sigma}_l(\cdot)$ for $l=1, \ldots, L-1$  is the ReLU (Rectified Linear Unit) function, which sets negative values to zero while keeping positive values unchanged, applied element-wise in multivariate settings.
In classification problems, for the final layer $L$, researchers often use a soft-max function \citep{fan2020selective}. 



As $\bfh(\cdot; \cdot)$ is determined by the neural network structure, the unknown parameter is $\bstheta$. To estimate $\bstheta$,
we minimize the empirical risk:
\begin{equation}\label{eq:loss}
    \mathcal{L}(\bstheta \mid  \mathcal{D}) = \frac{1}{N}\sum_{i=1}^N \CE(\bfy_i, \bfh(\bx_i;\bstheta)),
\end{equation}
where 
$\CE(\bfy_i, \bfh(\bx_i;\bstheta)) = - \sum_{k=1}^K y_{ik} \log\left( h_k(\bx;\bstheta) \right)$ is the cross-entropy loss between the probability vector $\bfh(\bx_i;\bstheta)$ and $\bfy_i$. 

\noindent\textbf{\textit{Knowledge Distillation and its Multi-Teacher Variant.}}
Knowledge distillation (KD), introduced by \citet{hinton2015distilling}, is a procedure where a simpler model (student) learns from a larger model (teacher). In the classical setup of KD, a pre-trained complex deep neural network model serves as the teacher. For each data point $i$ in the training set $\mathcal{D}$, the teacher model $M_t$ outputs a probability vector $\boldp_{i}=(p_{i1},\ldots p_{iK})^T$, where $p_{ij}$ represents the predicted class probability that $\bx_i$ belongs to class $j$. Mathematically, $p_{ij}$ approximates $\mathbb{P}(y_{ik} = 1 \mid \bx_i)$ according to model $M_t$.  
In this work, teacher models are treated as closed-source models, i.e., their structure and parameters are unknown.   This is common in practice as many popular pre-trained large language models, e.g., ChatGPT and Claudia, are closed-source.

In contrast, the student model $M_s$ is transparent, defined by a deep neural network with a known structure $\bfh(\cdot ;\cdot)$ 
and unknown parameter $\bstheta$. 
The student model $M_s$ is trained using both the training sample $\cD$ and the corresponding teacher model's predicted class probabilities $\boldp = \{\boldp_{i}\}_{i=1,\cdots,N}$.
The student model $M_s$ outputs a probability vector $\boldq_{i}=(q_{i1},\ldots q_{iK})^T$ for data point $i$, where $\boldq_{i} = \bfh(\bx_i;\bstheta)$ and $q_{ij}$ represents the predicted class probability that $\bx_i$ belongs to class $j$. 
KD measures the discrepancy between the student and teacher models with
\begin{equation}\label{eq:student_teach_discrepancy}
\begin{aligned}
\tilde{\mathcal{L}}( \bstheta \mid \mathcal{D}, \boldp ) &= \frac{1}{N}\sum_{i=1}^N \CE(\boldp_i, \bfh(\bx_i;\bstheta)) = - \frac{1}{N}\sum_{i=1}^N\sum_{k=1}^K p_{ik} \log\left( h_k(\bx_i;\bstheta) \right),
 \end{aligned}
\end{equation}
which is the sample mean of the cross-entropy $\CE(\boldp_i, \bfh(\bx_i;\bstheta))$ across the training data. 
To leverage the information from  both the training data  and the teacher model's predictions, KD  estimates the parameter $\bstheta$ by minimizing the objective function,
\begin{equation}\label{eq:opt_KD_loss}
\begin{aligned}
 \mathcal{L}^{\mathrm{KD}}( \bstheta \mid \mathcal{D}, \boldp, \lambda)  =  \mathcal{L}( \bstheta \mid \mathcal{D} ) + \lambda\tilde{\mathcal{L}}( \bstheta \mid \mathcal{D}, \boldp ) ,
\end{aligned}
\end{equation}
where  
$\mathcal{L}^{\mathrm{KD}}(  \bstheta \mid \mathcal{D}, \boldp, \lambda)$
combines the losses in Equation (\ref{eq:loss}) and Equation (\ref{eq:student_teach_discrepancy}) with a tuning parameter $\lambda \in (0,\infty)$  to balance their contributions.

In practice, a student can learn from multiple teachers, with the collective guidance of these teachers providing a comprehensive framework beneficial for training the student network. 
Consider $G$ pre-trained teacher models, denoted by $M_t^{(g)}$ for $g=1,\cdots,G$, we have the predicted class probabilities, denoted by $\boldp^{(g)} = \{\boldp_{i}^{(g)}\}_{i=1,\cdots,N}$, where $\boldp_i^{(g)}=(p_{i1}^{(g)},\ldots p_{iK}^{(g)})^T = M_t^{(g)}(\bx_i)$.
A natural way to formulate the multi-teacher KD loss is:
\begin{equation}\label{eq:opt_multi_KD_loss}
\begin{aligned}
 \mathcal{L}^{\operatorname{Multi-KD}}(  \bstheta \mid \mathcal{D}, \boldp, \lambda)  =  \mathcal{L}(  \bstheta \mid \mathcal{D} ) + \lambda 
     \sum_{g=1}^G w^{(g)}\tilde{\mathcal{L}}^{(g)}( \bstheta \mid \mathcal{D}, \boldp^{(g)} ),
\end{aligned}
\end{equation}
where 
$w^{(g)}\in (0,\infty)$ is the weight assigned to $M_t^{(g)}$, $\sum_{g=1}^G w^{(g)}=1$, and $\tilde{\mathcal{L}}^{(g)}( \bstheta \mid \mathcal{D}, \boldp^{(g)} )$ is analogously defined as in Equation (\ref{eq:student_teach_discrepancy}). Extensions, including assigning sample-wise weights and distilling intermediate-layer representations, have also been explored \citep{you2017learning, fukuda2017efficient,wu2019multi}.

While existing research has advanced our understanding of KD and its extensions to multi-teacher scenarios, significant gaps remain in capturing interactions among teachers and providing uncertainty quantification for the model predictions. In this work, we bridge these gaps by introducing a Bayesian framework for multi-teacher KD that explicitly models teacher-student interactions through a mixture prior, enabling both robust knowledge integration and uncertainty-aware learning. We provide details in the following section.

\section{Bayesian Knowledge Distillation}\label{sec:methodology}

Unlike the previous formulations of multi-teacher KD, in this work,  we develop multi-teacher KD in a Bayesian paradigm, which naturally enables uncertainty quantifications. 

\subsection{Bayesian Model with a Teacher-Informed Prior}

In our proposed Bayesian model, since each observation belongs to exactly one class,
we assume that the class label $\bfy_i$ follows a multinomial distribution given the predictor $\bx_i$,
\begin{equation}
\label{eq: model}
\bfy_i | \bx_i  \sim \text{Multinomial}(1, \boldq_i), 
\end{equation}
where $i=1, \ldots, N$, and the multinomial probability is provided by a deep neural network, i.e.,  $\boldq_i=\bfh(\bx_i;\bstheta)$. 

In the Bayesian paradigm, parameter $\bstheta$ is assumed to be random. 
Our goal is to develop a proper prior distribution. A straightforward approach might be to directly place a prior on $\bstheta$ as commonly did in Bayesian generalized linear models, see Chapter 14 in \cite{gelman2013bayesian}. However, in our study, the teacher models are close-source, meaning
the teacher models’ architectures and parameters are inaccessible, with only access to their outputs. Even in the case that the teacher models are open-source,  since the teacher and student models differ in structure, the teacher models’ knowledge cannot be directly encoded into a prior for the parameters of the student model.  Without any further information, the only option for us to impose objective (aka noninformative) prior for $\bstheta$, which is not satisfiable.
As an alternative, we propose using the teacher models’ predicted class probabilities to build the prior, referred to as the teacher-informed prior (TIP), for $\bstheta$. 
We set the following mixture distribution \citep{McLachlan2000} as our TIP for 
$\bstheta$, 
\begin{equation}\label{eq:prior_KD_theta}
 \pi(\bstheta\mid\{\boldp_{i}^{(g)}\}_{i=1,\cdots,N}^{g=1,\cdots,G}) 
    \propto \prod_{i=1}^N \sum_{g=1}^G  w_i^{(g)} f( \bfh(\bx_i; \bstheta)\mid\boldp_{i}^{(g)}),
\end{equation}
where $\propto$ denotes proportionality, $G$ is the number of mixture components, $w_i^{(g)}$ are the weights assigned to the component $g$ on data point $\bx_i$,
$f(\cdot\mid\boldp_{i}^{(g)})$ are probability density functions. 
In this work, we consider 
$f(\cdot\mid\boldp_{i}^{(g)})$
is the density function of a Dirichlet distribution $Dir(\mathbf{1}_K+\lambda \boldp_{i}^{(g)})$, where  $\lambda$ is a tuning parameter weighting the contribution of teacher model predictions. 
Hence, we have
\begin{equation}\label{eq:prior_KD_q}
 f(\boldq_i \mid\boldp_{i}^{(g)})  = \frac{1}{B(\mathbf{1}_K+\lambda \boldp_{i}^{(g)})}\prod_{k=1}^K (q_{ik})^{\lambda p_{ik}^{(g)}},
\end{equation}
where $B(\cdot)$ is the multivariate Beta function. Note that the mode of $f(\cdot \mid\boldp_{i}^{(g)})$ is $\boldp_{i}^{(g)}$ for $\lambda > 0$.  
Furthermore, we show that $\pi(\bstheta\mid\{\boldp_{i}^{(g)}\}_{i=1,\cdots,N}^{g=1,\cdots,G}) $ is a proper prior with the following proposition. We relegate the proof to Appendix \ref{supp:proof}.
\begin{proposition}\label{lemma:proper}
    Consider the probability density function $f(\boldq\mid\boldp_{i}^{(g)})$  as defined in Equation (\ref{eq:prior_KD_q}) with a constant $\lambda > 0$. 
    Assuming that the parameters of the student model lie in a compact space, then $\pi(\bstheta\mid\{\boldp_{i}^{(g)}\}_{i=1,\cdots,N}^{g=1,\cdots,G}) $ is a proper prior.
\end{proposition}

According to our best knowledge,  the proposed TIP is novel in the statistical literature.
The proposed TIP has four distinguishing features: (1) Broad knowledge integration. 
It is a prior integration of the information from both teacher models and task-specific training data. The teacher models are typically pre-trained using extraneous data, and our training data are fed into the teacher models to yield predicted probabilities. Therefore, it captures extensive external knowledge while maintaining direct relevance to the task.
(2) Statistical robustness. It reduces overfitting risk by incorporating external information, making it more robust than purely data-driven priors used in classical Bayesian analysis. This is particularly advantageous with limited training data.
(3) Adaptability. It offers adaptability by incorporating multiple teacher models with automatic weight adjustment based on model quality, and the prior can be updated when better teacher models become available.
(4) Practically user-friendly. It relies less on expert elicitation and is scalable to different domains, including texts, images, and numerical data. 
The key innovation is the principled integration of external knowledge (LLMs) with empirical data, offering better generalization than either approach alone.

In the general case with multiple teachers, the posterior density of $\bstheta$ is,
\begin{small}
\begin{align}\label{eq:post}
    \pi(\bstheta \mid\mathcal{D}, \boldp, \lambda,  \bfh(\cdot;\cdot) ) 
    & \propto 
    \exp\left\{ \sum_{i=1}^N \sum_{k=1}^K y_{i k} \log \left(h_k\left(\bx_i; \bstheta \right)\right) \right. \nonumber\\
    & \left. + \sum_{i=1}^N \log \left[\sum_{g=1}^G w_i^{(g)} \frac{1}{B\left(\mathbf{1}_k+\lambda \boldp_i^{(g)}\right)} \prod_{k=1}^K\left(h_k\left(\bx_i; \bstheta\right)\right)^{\lambda p_{ik}^{(g)}}\right] \right\}.
\end{align}
\end{small}
Please note that when $G=1$, Equation (\ref{eq:post}) reduces to a single-teacher case \citep{fangbayesian}. We present the discussion of this special case in Appendix \ref{supp:add_single}.

\begin{figure}[t]
\vskip 0.2in
\vspace{-16pt}
\centering
\includegraphics[width=0.95\linewidth]{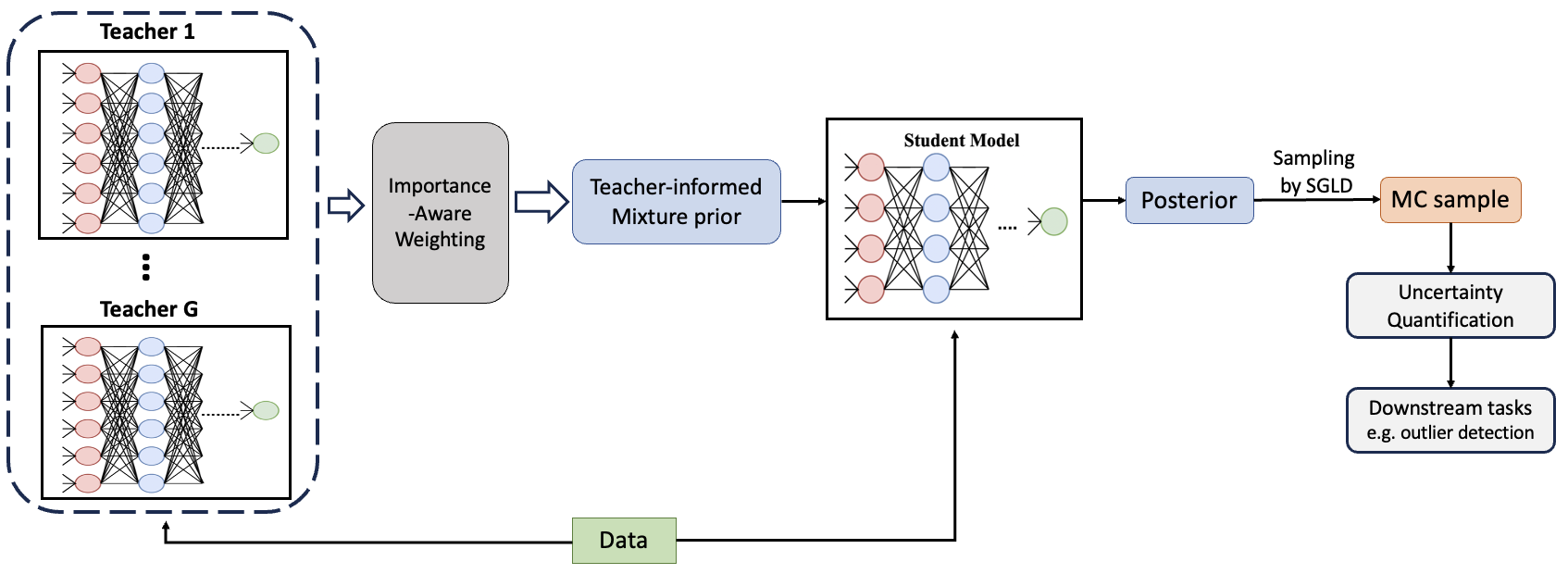}
\vspace{-0.1cm}
\caption{\textbf{The multiple teacher Bayesian knowledge distillation (MT-BKD) framework.} A teacher-informed prior is established for the student model's parameters based on the predicted probabilities from multiple teacher models, and the posterior distribution is derived. An importance-aware weighting mechanism balances contributions from the teachers. The stochastic Gradient Langevin Dynamics (SGLD) method is then applied to generate Monte Carlo samples from the posterior. Uncertainty quantification of the predictions and the subsequent downstream tasks, such as outlier detection, can be achieved accordingly.} \label{fig:flowchart}
\vskip -0.1in
\end{figure}

Applying the mixture prior defined in Equation (\ref{eq:prior_KD_theta}) leads to improved results compared to the traditional multi-teacher knowledge distillation (KD) strategy described in Equation (\ref{eq:opt_multi_KD_loss}). Specifically, the traditional multi-teacher KD strategy is equivalent to independently deriving posterior distributions from each teacher model and combining them using logarithmic opinion pooling \citep{genest1986characterization, garthwaite2005statistical}. We establish the following theorem, with the proof presented in Appendix \ref{supp:proof}.


\begin{theorem}\label{thm:equiv_multi}
    Given the prior distribution defined in Equation (\ref{eq:prior_KD_theta}) and Equation (\ref{eq:prior_KD_q}), the traditional multi-teacher KD strategy in Equation (\ref{eq:opt_multi_KD_loss}) is equivalent to deriving each teacher model’s posterior independently and combining them using logarithmic opinion pooling.
\end{theorem}
Traditional multi-teacher KD algorithms typically provide only point estimates, overlooking the uncertainty in model parameters. In contrast, our proposed MT-$\ours$ method estimates the full posterior distribution, equipping the student model with uncertainty quantification capabilities. This is particularly important when the posterior distribution is multimodal. Point estimates in such cases may fall between modes, failing to capture the true distribution and potentially leading to misleading conclusions.
In addition, we employ a mixture prior approach in MT-$\ours$. We begin with a mixture prior with fixed weights representing our initial beliefs about each component. Once data is observed, we compute the posterior distribution based on this mixture prior. Although the initial prior weights remain fixed, the posterior naturally shifts in favor of components that better explain the data, balancing our assumptions with empirical evidence. If multiple priors fit the data well, the posterior may be multimodal, reflecting multiple plausible parameter configurations. By maintaining a mixture of priors, we preserve these distinct possibilities while allowing the data to emphasize those that are most consistent.
Conversely, combining individual posteriors via logarithmic opinion pooling first updates each prior with data to obtain individual posteriors and then combines them via a weighted geometric mean of their log-transformed values. As a result, each prior retains its predetermined weight in the final combined posterior, regardless of its performance on the data.


\begin{figure}[h]
\vskip 0.2in
\vspace{-16pt}
\centering
\includegraphics[width=0.6\linewidth]{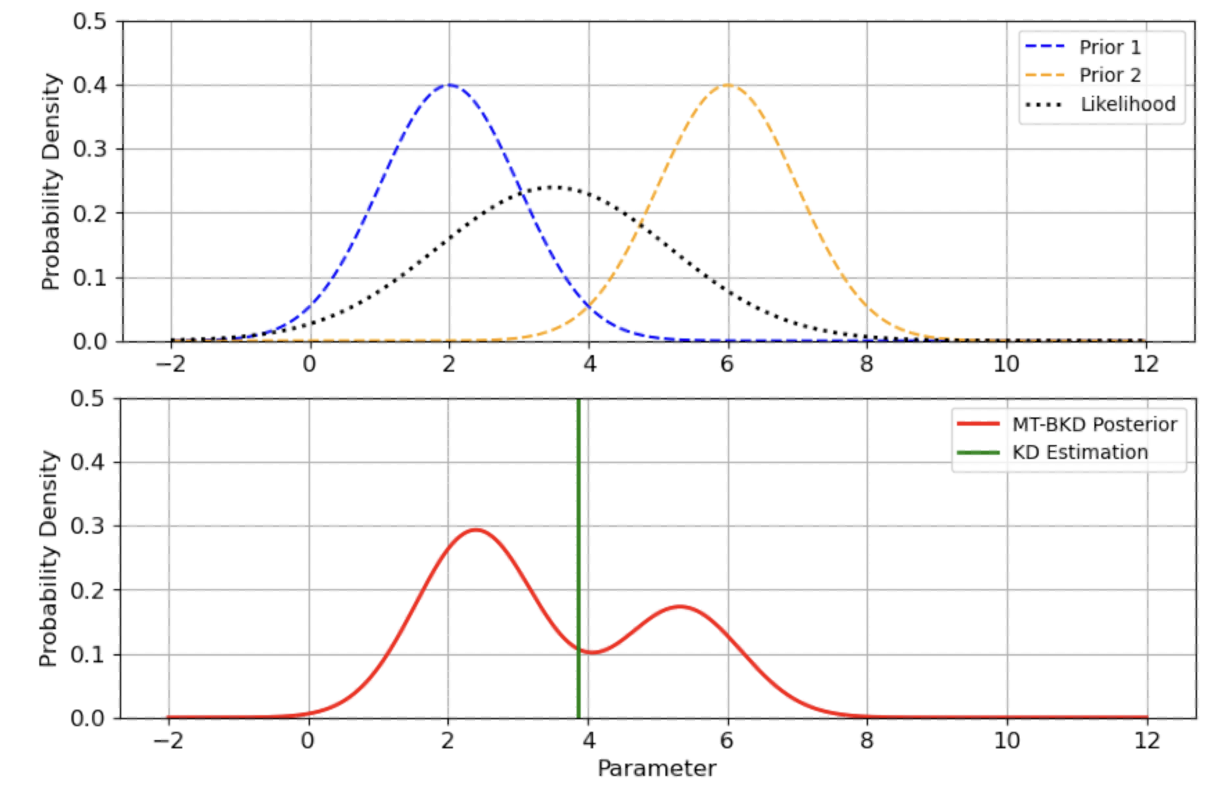}
\vspace{-0.1cm}
\caption{Comparison of posterior distributions obtained through MT-$\ours$ and the estimate of multi-teacher KD. The figure illustrates the prior distributions, the data distribution, the resulting posterior distributions derived from MT-BKD (red line), and the estimation achieved through multi-teacher KD (green line).} \label{fig:toy_mix}
\vskip -0.1in
\end{figure}

We illustrate the advantage of MT-BKD using a toy example with two Gaussian priors centered at $\mu_1=2.0$ and $\mu_2=6.0$ ($\sigma^2=1.0$). The observed data has a mean of $3.5$, which lies between these priors, indicating no strong preference for either. Figure \ref{fig:toy_mix} shows that MTBKD (red line) yields a bimodal posterior, accurately reflecting the balanced support for both priors while slightly favoring the one better supported by the data. In contrast, traditional multi-teacher KD (green line) collapses to a point estimate, overlooking the bimodal nature implied by the priors and the data.

\subsection{Importance-Aware Distillation}\label{sec:weight}

Equation (\ref{eq:prior_KD_theta}) highlights the critical role of the weights $w_i^{(g)}$ in distilling knowledge from multiple teacher models into a unified student model. Since teacher networks vary in their effectiveness across different data, identifying and leveraging these differences is essential for successfully integrating multiple teachers.
We define $w_i^{(g)}$ for the teacher model $M_t^{(g)}$ as
\begin{align}\label{eq:weight}
    w_i^{(g)} & = 
    \frac{ \xi(H(\boldp_i^{(g)})) }{ \sum_{g=1}^G \xi(H(\boldp_i^{(g)})) } ,
\end{align}
where $H(\cdot)$ represents the entropy, calculated as $H(\boldp_i^{(g)}) = -\sum_{k} p_{ik}^{(g)} \log(p_{ik}^{(g)})$, and $\xi$ is a monotonically descreasing function. 
Intuitively, from the perspective of information theory, information entropy is used to describe the degree of randomness and uncertainty of information, so it is reasonable to apply information entropy to quantify the knowledge scale of the teacher model on the given data point.
Thus, we assign higher distillation weights to teacher models with lower entropy predictions. This allows the student model to learn more strongly from the teachers that have developed greater expertise on each specific data point, leading to more effective knowledge transfer.


We now discuss the effect of $\lambda$.
In the extreme case of $\lambda = 0$, $f(\boldq\mid\boldp_{i}^{(g)})$ becomes the density of a symmetric Dirichlet distribution, $Dir(\mathbf{1}_K)$. This implies that there is no prior knowledge favoring one category over another when classifying each data point.
For the case of $\lambda \to \infty$, we derive the following proposition with proof in Appendix \ref{supp:proof},
\begin{proposition}\label{thm:converge}
    Consider the probability density function $f(\boldq\mid\boldp_{i}^{(g)})$ as defined in Equation (\ref{eq:prior_KD_q}), as $\lambda \to \infty$, we have
    $f(\boldq \mid\boldp_{i}^{(g)}) \longrightarrow \delta(\boldq-\boldp_{i}^{(g)}) ,$
    where $\delta(\cdot)$ is the multivariate Dirac delta function.
\end{proposition}
The above proposition shows that, as $\lambda$ approaches infinity, the student model assigns nonzero probability only when its predicted class probabilities exactly match the teacher’s, aligning with original knowledge distillation (KD) that fully trusts the teacher's probabilities at infinitely large $\lambda$.
In summary, a large $\lambda$ indicates strong reliance on the teacher's predictions, while a small $\lambda$ suggests limited teacher influence. 
In practice, $\lambda$ can be determined by selecting candidate values based on domain expertise and selecting the best one through cross-validation \citep{shao1993linear}.

\section{Uncertainty Quantification via Bayesian Predictive Inference}\label{sec:uncertainty}

In practice, it is very common to use the trained student model to analyze confidential data.  Analyzing confidential data using pre-trained teacher models as presented in the previous section may result in information leaking \citep{ray2023samsung}. 
In this section, we develop a method to quantify prediction uncertainty using the trained student model alone.

\subsection{Measurement of Prediction Uncertainty}
For a new data point $\tilde{\bx}$, we denote the trained student model's prediction as $\tilde{\boldq}=\bfh(\tilde{\bx}; \bstheta) $. 
Suppose that the label of $\tilde{\bx}$ is $\tilde{\bfy}$. Instead of estimating the predictive posterior distribution \citep{gelman1995bayesian} of the unobserved label $\tilde{\bfy}$, we introduce a single-valued metric to gauge our confidence in the model’s prediction.

A commonly used confidence measure for classification is the deviance \citep{spiegelhalter2014deviance}, defined as $\operatorname{dev}\left(\tilde\bfy, \tilde\boldq\right) = -2 \sum_{k=1}^K \tilde{y}_{k} \log(\tilde{q}_{k}) $, where $\tilde{y}_{k}$ and $\tilde{q}_{k}$ are the $k$th elements of $\tilde\bfy$ and $\tilde\boldq$.
This deviance can be interpreted as a generalized sum of squared residuals for $\tilde\boldq$. A higher deviance value suggests greater uncertainty in the model's predictions.
To make predictive inferences about $\operatorname{dev}\left(\tilde\bfy, \tilde\boldq\right)$, we use its posterior predictive distribution.
Under the multinomial model in Equation (\ref{eq: model}) and the posterior distribution of $\bstheta$ in Equation (\ref{eq:post}), we have
\begin{equation}\label{eq: post_pred_dev}
    \pi \left(\operatorname{dev}\left(\tilde\bfy, \tilde\boldq\right) | \mathcal{D}, \boldp, \lambda,  \bfh(\cdot;\cdot)   \right) =  \int \pi( \bstheta | \mathcal{D}, \boldp, \lambda,  \bfh(\cdot;\cdot)   )   p\left(\operatorname{dev}\left(\tilde\bfy, \tilde\boldq\right) | \bstheta \right) d \bstheta,
\end{equation}
where $p\left(\operatorname{dev}\left(\tilde\bfy, \tilde\boldq\right) | \bstheta \right)$ is the conditional probability density function of $\operatorname{dev}\left(\tilde\bfy, \tilde\boldq\right)$ given $\bstheta$. 
Finally, the posterior mean deviance is used to quantify the model’s prediction uncertainty.

\subsection{Posterior Sampling}
For the high dimension of the parameter $\bstheta$ and the complex structure of $\bfh(\cdot ; \bstheta)$, it may not be feasible to obtain an analytical solution for Equation (\ref{eq: post_pred_dev}).
Hence, we apply the stochastic Gradient Langevin Dynamics (SGLD) \citep{welling2011bayesian}, a variant of Langevin Monte Carlo, 
to efficiently sample from the complex, high-dimensional posterior distributions.
SGLD achieves this by seamlessly integrating the general approach of stochastic gradient descent with Langevin dynamics. 

Suppose $\bstheta^{\langle j-1 \rangle}$ is sampled in iteration $j-1$. Then in the $j$-th step, given a mini-batch of $m$ data points $\mathcal{D}^{\langle j \rangle} = \{(\bx_i^{\langle j \rangle}, \bfy_i^{\langle j \rangle})\}_{i=1,\cdots,m}$ and the class probability $\{\boldp_i^{(g)\langle j \rangle}\}_{i=1,\cdots,m}^{g=1,\cdots,G}$ predicted by the teacher model, SGLD generates the sample from posterior using gradient updates plus Gaussian noise,
\begin{small}
\begin{align}\label{eq:SGLD}
    &\bstheta^{\langle j \rangle}\nonumber
    = \nonumber \bstheta^{\langle j-1 \rangle}+\tau \nabla_{\bstheta} log \pi\left(\bstheta^{(j-1)} \mid \mathcal{D}^{\langle j \rangle},\{\boldp_i^{(g)\langle j \rangle}\}_{i=1,\cdots,m}^{g=1,\cdots,G} , \lambda,  \bfh \right)  +  \sqrt{2\tau}  \xi^{\langle j \rangle},
\end{align}    
\end{small}
where $\tau$ is the step size, and $\xi^{\langle j \rangle} $ is randomly sampled from $ N(0, I)$. 
SGD optimizes the log-likelihood to guide the sampling process toward regions of higher probability, while Langevin dynamics introduces controlled noise into the parameter updates, ensuring convergence to the full posterior distribution rather than merely its mode.
Specifically, in the limit of $j \rightarrow \infty$ and $\tau \to 0$, the probability distribution of $\bstheta^{\langle j \rangle}$, denoted as $\rho^{\langle j \rangle}$, converges to the posterior distribution of $\bstheta$.
By using gradient information and introducing controlled noise, SGLD becomes more efficient in handling high-dimensional data \citep{girolami2011riemann}. Furthermore, SGLD's use of mini-batches for gradient computation alleviates the computational burden associated with optimizing over entire datasets, making it particularly well-suited for large-scale datasets.


\begin{algorithm}[tb]
   \caption{Multi-Teacher Bayesian Knowledge Distillation (MT-BKD).}
   \label{algo:BKD}
\begin{algorithmic}
   \State {\bfseries Input:} $\mathcal{D}=\{(\bx_i, \bfy_i)\}_{i=1}^{N}$,$\bfh(\cdot;\cdot)$, $\tau$, $\lambda$, $r$. New dataset $\mathcal{T}=\{(\tilde \bx_s,\tilde \bfy_s)\}_{s=1}^{S}$
   \State 1. Get the output $\boldp^{(g)},\ g=1,\cdots,G$, of each teacher model for each data point in $\mathcal{D}$.
   \State 2. Build the mix prior as in Equation (\ref{eq:prior_KD_theta}) using weights defined in Equation (\ref{eq:weight}).
    \State 3. Calculate the posterior distribution as in Equation (\ref{eq:post}).
    \State 4. Generate Monte Carlo sample of $\bstheta$:
    \vspace{-6pt}
    \begin{itemize}
    \item[$\boldsymbol{\cdot}$] At iteration $j^{th}$ with a subset of $m$ data points $\mathcal{D}^{\langle j \rangle} = \{(\bx_i^{\langle j \rangle}, \bfy_i^{\langle j \rangle})\}_{i=1}^{m}$,
        \begin{itemize}
        \vspace{-3pt}
        \item[$\cdot$] Sample $\xi^{\langle j \rangle} \sim N(0, I)$,
        \vspace{-3pt}
        \item[$\cdot$] Sample $\boldsymbol{\theta}^{\langle j \rangle}$ using SGLD.
        \vspace{-6pt}
        \end{itemize}
    \end{itemize}
    \State 5. Predict the uncertainty $\{ \tilde{\Delta}_s \}_{s=1}^S$ for each data point in $\mathcal{T}$ using Equation (\ref{eq:mean_dev}). 
    \State 6. Construct credible intervals $\{ [0,\tau_s] \}_{s=1}^S$ for the deviance of the prediction on $\mathcal{T}$.
    \State {\bfseries Output:} $\{ \bstheta^{\langle j \rangle} \}_{j=1}^r$, $\{ \tilde{\Delta}_s \}_{s=1}^S$, and $\{ [0,\tau_s] \}_{s=1}^S$.
\end{algorithmic}
\end{algorithm}

Since the distribution of MC sample $\bstheta^{\langle j \rangle}$ converges to the posterior distribution of $\bstheta$ as $j\to \infty$, it allows for a precise estimation of the characteristics of $\bstheta$. Consequently, it fosters a comprehensive analysis and understanding of our model covering various aspects. One aspect that we are particularly interested in is the model's prediction on new data. 

\subsection{Posterior Mean Deviance and Credible Interval}
With the posterior predictive distribution of the deviance in Equation (\ref{eq: post_pred_dev}), we propose the posterior mean deviance as a measure of predictive uncertainty and provide the credible interval of deviance.


With the MC sample $\{  \bstheta^{\langle j \rangle} \}_{j=1}^r$, we can estimate the 
posterior mean of $\operatorname{dev}\left(\tilde\bfy, \tilde\boldq\right)$ by, 
\begin{small}
\begin{equation}
\begin{aligned}\label{eq:mean_dev}
    \tilde{\Delta} &=  \frac{1}{r}\sum_{j=1}^{r} E\left( \operatorname{dev}(\tilde \bfy,\tilde{\boldq}) \mid \bstheta^{\langle j \rangle} \right)\nonumber\\
    &= -\frac{2}{r}\sum_{j=1}^{r} \sum_{k=1}^K \tilde{q}_{k}^{\langle j \rangle} \log(\tilde{q}_{k}^{\langle j \rangle}) ,
\end{aligned}
\end{equation}
\end{small}
where $\tilde{q}_{k}^{\langle j \rangle} $ is the $k$th element of $\tilde{\boldq}^{\langle j \rangle} = \bfh(\tilde\bx;\bstheta^{\langle j \rangle})$.


Now, we construct the credible interval for the deviance of the prediction.
For the deviance $\operatorname{dev}\left(\tilde\bfy, \tilde\boldq\right)$, its $1-\alpha$ credible interval,  denoted by $CI = [0,\tau]$, satisfies
\begin{small}
\begin{equation}\label{eq:CI}
    E(  \textbf{1}_{\{ \operatorname{dev}\left(\tilde\bfy, \tilde\boldq\right) \leq \tau \}} ) = 1-\alpha ,
\end{equation}
\end{small}
where $\alpha$ is the credible level.
We can empirically estimate $\tau$ by minimizing $L(\tau)$ where
\begin{small}
    \begin{align}\label{eq:empirical}
        L(\tau) =& \left| \frac{1}{r} \sum_{j=1}^{r} E(  \textbf{1}_{\{ \operatorname{dev}\left(\tilde\bfy, \tilde\boldq\right) \leq \tau \}}  \mid \bstheta^{\langle j \rangle} )-(1-\alpha) \right| \nonumber \\
       =&
      \left| \frac{1}{r} \sum_{j=1}^{r}\sum_{k=1}^K   \tilde{q}_k^{\langle j \rangle}  \textbf{1}_{\{ -2\log(\tilde{q}_{k}^{\langle j \rangle}) \leq \tau \}}  -(1-\alpha) \right|. 
    \end{align}
\end{small}
The function $L(\tau)$ denotes the difference between the left-hand side and the right-hand side of Equation (\ref{eq:CI}).

Constructing a credible interval on deviance, rather than directly on the prediction $\boldq$, presents a notable advantage. It enables the calculation of the coverage rate of the constructed interval using the true label. This is beneficial since, in most cases, the true labels of the testing dataset are readily available for evaluation, whereas the true class probabilities are typically elusive.
The MT-BKD algorithm is summarized in Algorithm \ref{algo:BKD}.




\section{Theoretical Analysis}
Here, we develop an asymptotic convergence of our proposed MT-BKD estimate.
Without the loss of generality, we only consider the binary classification problem here.
Let $\mu$ represent the predictor $\bx$'s true density function in the training sample of the student model while $\mu^{(g)}$ represents the predictor's true density function in the training sample of the $g$th teacher model for $g = 1, \ldots, G$.
Suppose the true parameter for the student model defined by Equation(\ref{dnn_def}) and (\ref{eq: model}) is $\bstheta_0$.
To define a good estimate, we consider the excess $\phi$-risk with the log loss,
\begin{equation}\label{eq:excess phi risk}
    \cE^\phi(\hat{\bstheta}_n,\bstheta_0) =  \cE^\phi(\hat{\bstheta}_n)-\cE^\phi(\bstheta_0),
\end{equation}
where $ \cE^\phi(\bstheta) = E_{\bx,\bfy}[\phi(\bfy^T \logit [\bfh(\bx;\bstheta)] )]$, and $\phi(x)= \log(1 + \exp(−x))$.

We define the quadratic functional $V(f) = \int_{\Omega_\bx} f^2 \mu(\bx) d\bx$.

The following conditions are assumed to hold.

(A1) For any $\bx \in \Omega_\bx$,  $\mu(\bx)>0$, $\max_{ 1\leq g \leq G} \mu^{(g)}(\bx) >0$. 

(A2) The minimizer $\bstheta^*_n$ in Equation (\ref{eq:loss}) yield a consistent estimate of $\bfh(\bx;\bstheta)$, i.e., the error $\int_{\Omega_\bx} [\bfh(\bx;\bstheta^*_n) - \bfh(\bx;\bstheta) ]^2 \mu(\bx) d\bx  = O_p(C_1(n)) $ and $\lim_{n \to \infty } C_1(n) = 0$. 

(A3) There exist a  positive constant $\varepsilon^{(g)}$ such that $\int_{\Omega_\bx} [M_t^{(g)}(\bx) - \bfh(\bx;\bstheta) ]^2 \mu^{(g)}(\bx) d\bx \leq \varepsilon^{(g)} $ for $g = 1,\ldots,G$.


Condition (A1) essentially requires that for any $\bx \in \Omega_\bx$, there exists at least one of the training samples' population distributions corresponding to the teacher models that assign a nonzero density to it.
Condition (A2) can be proven under several regularity conditions, e.g., $\bfh(\bx;\bstheta)$ is a $(p,C)$-smooth generalized hierarchical interaction model \citep{bauer2019deep}. Without loss of generality, we use $C_1(n)$ to represent its convergence rate while the detailed forms under different scenarios can be found in the existing literature \citep{gyorfi2006distribution, horowitz2007rate,bauer2019deep,schmidt2020nonparametric}. 
Condition (A3) requires that the estimation errors of the teacher models be finite.

\begin{theorem}[Convergence of MT-BKD]\label{thm:asymp}
The estimate of $\bstheta$ by maximizing Equation(\ref{eq:opt_multi_KD_loss}) yields a consistent estimate of $\bfh(\bx;\bstheta)$ with the optimal convergence rate upper bounded by 
\begin{equation}\label{eq:error rate}
    C_2(n) \wedge   G\cdot V \left( \left\{\sum_{g=1}^G \left[\frac{\mu^{(g)}(\bx)}{\mu(\bx) \varepsilon^{(g)} } \right]^{\frac{1}{2}} \right\}^{-1}\right)  ,
\end{equation}
where $C_2(n) \asymp C_1(n)$.
\end{theorem}
We find the estimation error of MT-BKD could be potentially reduced compared to if the second term of Equation (\ref{eq:error rate}) is smaller than $C_2(n)$. Smaller $\varepsilon^{(g)}$ leads to a smaller value of the second term.

\noindent\textbf{Remark}. If $\mu^{(g)} = \mu$ for $g = 1,\ldots,G$,   the second term of Equation (\ref{eq:error rate}) equals to $G\left(\sum_{g=1}^G \sqrt{1/\varepsilon^{(g)}}  \right)^{-2} $. 


\section{Simulation}\label{sec:simul}

We conduct simulation studies to evaluate the performance of the proposed MT-BKD method in classification tasks under two scenarios with varying dimensions and class counts.
Model performance is assessed using classification accuracy and mean square error (MSE) on class probability $p$, where MSE for MT-BKD compares the predicted posterior mode with the true $p$. Beyond overall accuracy, we focus on the model's ability to quantify uncertainty. Mean deviance, as discussed in Section \ref{sec:uncertainty}, is used to measure prediction uncertainty, and its relationship with the original data is examined. We further analyze coverage rates across different confidence levels and sample sizes, offering a comprehensive evaluation of the method's uncertainty quantification.


We compare MT-BKD with four benchmark methods: (1) the teacher model; (2) the original KD method; (3) a combination of original KD with a Bayesian neural network (BNN) \citep{blundell2015weight}; and (4) a combination of original KD with Monte Carlo dropout (Dropout) \citep{gal2016dropout}. Details on the implementation of BNN and Dropout are provided in Appendix \ref{supp:simu}. For each distillation method, we evaluate both weighted distillation across multiple teachers, as described in Section \ref{sec:weight}, and distillation with equal teacher weights.


\subsection{Simulation 1}\label{sec:simu1}
In this simulation, we consider binary classification problems. Let $p(\bx) = \frac{\exp(\eta(\bx))}{1+\exp(\eta(\bx))}$, and $Y\sim Bernoulli(p(\bx))$
with $\eta(\bx) = -0.5+0.1x_{1}+0.8x_{1}^2 - 0.8x_{2}^2.$
We generate a sample of size $20,000$ by letting $x_{1}\sim Unif(-3,3.5)$ and $x_{2}\sim Unif(-3,3)$. The top left panel of Fig. \ref{fig:s1_mean_dev} shows the probability distribution $p(\bx)$ mapped across $x_1$ and $x_2$ coordinates, with probability values indicated by color intensity.

\begin{figure}[htp]
\centering
\includegraphics[width=0.9\linewidth]{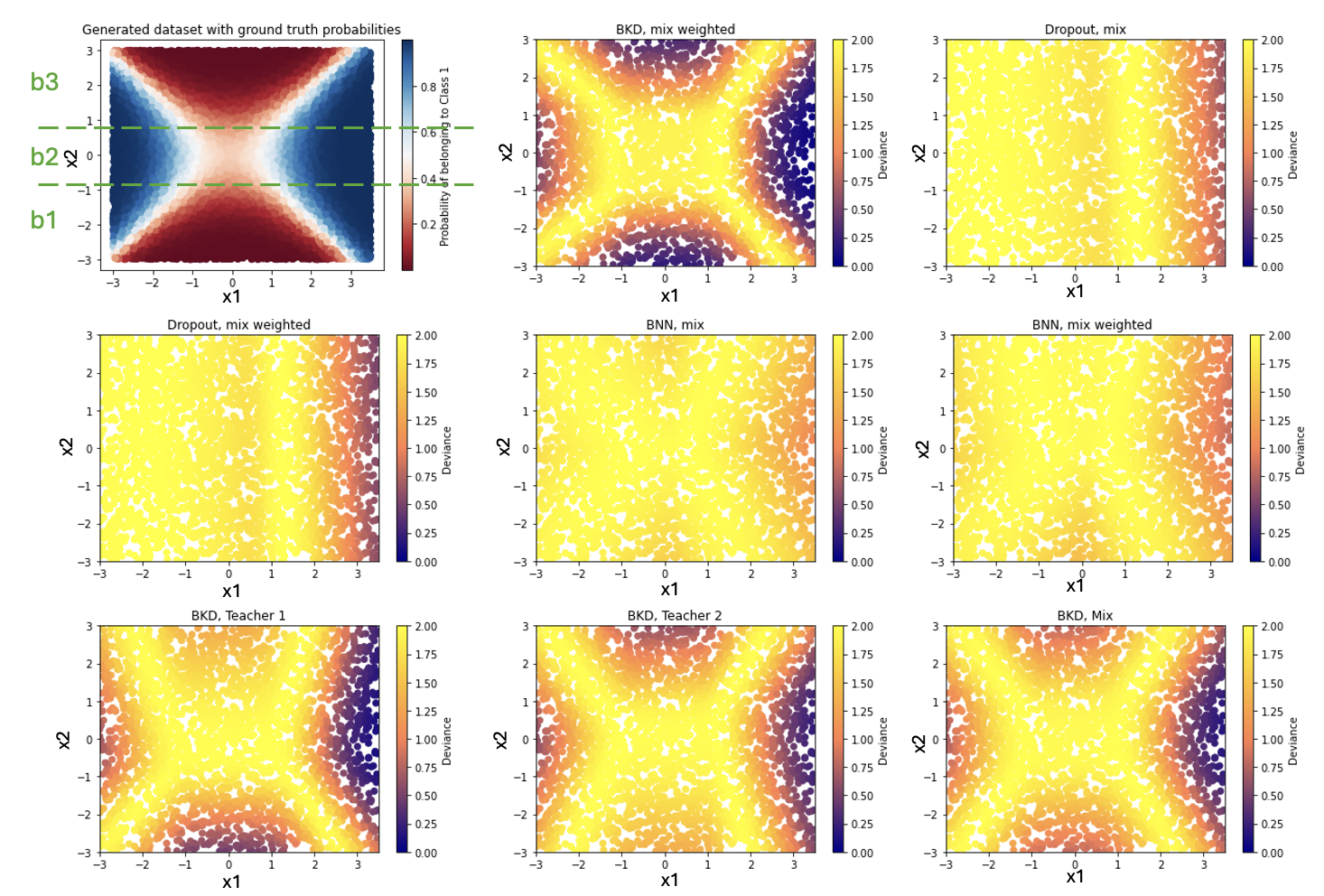}
\caption{Top left panel: Ground truth probability distribution $p(\bx)$ with three domains marked as $b_1$, $b_2$ and $b_3$. Remaining panels: Predicted uncertainty (measured as deviance) for each data point across various methods. } \label{fig:s1_mean_dev}
\end{figure}

To simulate teacher specialization, we divide the dataset into three domains $b_1$, $b_2$, and $b_3$, as shown in the top left panel of Fig. \ref{fig:s1_mean_dev}. We have Teacher 1 primarily trains on domain $b_1$, with limited exposure to the other domains, while Teacher 2 specializes in $b_3$. As shown in Table \ref{tab:acc_s1}, Teacher 1 excels in $b_1$, whereas Teacher 2 achieves high accuracy in $b_3$. We then train a student model on data evenly drawn from all three domains, distilling knowledge from both teachers. Details on model structure are in Appendix \ref{supp:simu}.


\noindent \textbf{Accuracy and MSE.} Table \ref{tab:acc_s1} presents the overall and domain-specific classification accuracy of all methods.
We have the following findings: (1) Despite their smaller model sizes, student models trained with KD methods achieve performance comparable to or even exceeding that of their teacher models. (2) MT-BKD outperforms original KD, Dropout, and BNN methods, achieving the highest overall accuracy (0.83) and excelling particularly in domains $b_1$ (0.87) and $b_3$ (0.85). (3) The proposed weighted distillation strategy consistently surpasses the equal-weight approach, highlighting its effectiveness in enhancing model performance. 
Additionally, MSE results for class probability $p$ demonstrate that the proposed MT-BKD method outperforms all other methods, with the weighted version achieving the lowest MSE.

\begin{table}
\centering
\scalebox{0.85}{
\begin{tabular}{c|cc|cc|cc|cc|cc}
\hline
& \multicolumn{2}{c|}{Teachers} & \multicolumn{2}{c|}{Orig KD} & \multicolumn{2}{c|}{MT-BKD} & \multicolumn{2}{c|}{Dropout} & \multicolumn{2}{c}{BNN} \\
Acc & T1 & T2 & equal & weighted & equal & weighted & equal & weighted & equal & weighted \\
\hline
Overall & 0.66 & 0.67 & 0.72 & 0.72 & 0.79 & \textbf{0.83} & 0.72 & 0.73 & 0.77 & 0.73 \\
$b_1$ & 0.84 & 0.41 & 0.77 & 0.75 & 0.83 & \textbf{0.87} & 0.75 & 0.76 & 0.78 & 0.78 \\
$b_2$ & 0.78 & \textbf{0.79} & 0.66 & 0.69 & 0.78 & 0.78 & 0.67 & 0.68 & 0.80 & 0.80 \\
$b_3$ & 0.39 & 0.83 & 0.74 & 0.74 & 0.79 & \textbf{0.85} & 0.73 & 0.73 & 0.76 & 0.63 \\
\hline
MSE & 0.14 & 0.14 & 0.08 & 0.08 & 0.02 & \textbf{0.01} & 0.08 & 0.08 & 0.07 & 0.07 \\
\hline
\end{tabular} }
\caption{Accuracy and MSE comparison for Simulation 1. We evaluate models on: overall accuracy, accuracy on data from each domain $b_i$ with $i=1,2,3$, and overall MSE for class probability $p$. The best results in each scenario are highlighted in bold.}
\label{tab:acc_s1}
\end{table}


\noindent \textbf{Uncertainty quantification.} 
We evaluate the inference performance of MT-BKD and compare it to Dropout and BNN. For clarity, we also include results from the single-teacher BKD method to assess individual teacher contributions.
Fig. \ref{fig:s1_mean_dev} compares the uncertainty quantification performance of various methods, with the top left subplot displaying the ground truth probability distribution.
In this setting, where $Y\sim Bernoulli(p)$, the variance of each observation is highly correlated with $p(1-p)$, making the class probability $p$ a key determinant of uncertainty.  Predictions for data points with $p$ near 0.5 are inherently more uncertain, while predictions for $p$ close to 0 or 1 are more certain.
Consequently, the white region ($p\approx 0.5$) in the top left subplot should exhibit the highest uncertainty, whereas the dark blue/red regions ($p\approx 0/1$) should show the lowest. Among the methods, weighted MT-BKD (top middle) aligns best with this expected uncertainty pattern.

The bottom-left subplots illustrate BKD's performance with individual teachers. The student model trained solely with Teacher 1 exhibits clear uncertainty contrasts (yellow vs. dark purple) in domain $b_1$ but blurred boundaries in $b_3$. In contrast, the model trained with Teacher 2 shows sharp uncertainty patterns in $b_3$ but weak distinction in $b_1$, reflecting each teacher's domain expertise. The weighted MT-BKD effectively integrates both teachers' strengths, resulting in improved uncertainty estimation across all domains.

\noindent \textbf{Coverage rate.} 
With the credible interval constructed using Equation (\ref{eq:empirical}), 
We evaluate the coverage rate at three standard credible levels (0.85, 0.90, 0.95). 
As shown in Fig. \ref{fig:sim_coverage}, the weighted MT-BKD method achieves strong performance, with empirical coverage closely matching the nominal credible levels, indicating its effectiveness in quantifying uncertainty. A strong competitor is the weighted BNN method, which achieves coverage performance comparable to the weighted MT-BKD method.

\subsection{Simulation 2}

In this simulation, we consider the following multi-class scenario. Let $Y \sim Multinomial(\boldp(\bx))$, and $ \boldp(\bx) = (p_1(\bx), p_2(\bx), p_3(\bx), p_4(\bx), p_5(\bx))^T$, where $ p_i(\bx) = \frac{f(\bx|\boldsymbol{\mu}_i,\boldsymbol{\Sigma})}{\sum_{j=1}^5f(\bx|\boldsymbol{\mu}_1,\boldsymbol{\Sigma})},\ i=1,\ldots,5$.
We have $\boldsymbol{\mu}_1=(0,0,0,0,0)^T$, $\boldsymbol{\mu}_2=(3,3,3,3,3)^T$, $\boldsymbol{\mu}_3=(0,0,1,3,2)^T$, $\boldsymbol{\mu}_4=(2,0,1,2,1)^T$, $\boldsymbol{\mu}_5=(2,2,1,0,1)^T$, and $\Sigma=(0.5^{|i-j|})_{5\times5}$. 
We simulate $8,000$ sample from each $N(\boldsymbol{\mu}_i,\Sigma)$, $i=1,\ldots,5$, and generate labels accordingly.

\begin{table}[t]
\centering
\scalebox{0.85}{
\begin{tabular}{c|ccc|cc|cc|cc|cc}
\hline
& \multicolumn{3}{c|}{Teachers} & \multicolumn{2}{c|}{Orig KD} & \multicolumn{2}{c|}{MT-BKD} & \multicolumn{2}{c|}{Dropout} & \multicolumn{2}{c}{BNN} \\
Acc & T1 & T2 & T3 & equal & weighted & equal & weighted & equal & weighted & equal & weighted \\
\hline
Overall & 0.56 & 0.61 & 0.53 & 0.85 & 0.85 & 0.85 & \textbf{0.87} & 0.86 & 0.86 & 0.84 & 0.85 \\
$b_1$ & \textbf{0.96} & 0.00 & 0.33 & 0.84 & 0.84 & 0.91 & 0.92 & 0.83 & 0.83 & 0.77 & 0.79 \\
$b_2$ & 0.90 & \textbf{0.92} & 0.05 & 0.90 & 0.90 & 0.89 & 0.89 & 0.92 & \textbf{0.92} & 0.91 & 0.90 \\
$b_3$ & 0.35 & 0.91 & \textbf{0.93} & 0.90 & 0.91 & 0.86 & 0.88 & 0.91 & 0.92 & 0.90 & 0.90 \\
$b_4$ & 0.58 & 0.53 & 0.51 & 0.79 & 0.78 & \textbf{0.86} & \textbf{0.86} & 0.81 & 0.80 & 0.81 & 0.82 \\
$b_5$ & 0.00 & 0.67 & 0.83 & \textbf{0.85} & 0.84 & 0.76 & 0.82 & 0.82 & 0.81 & 0.81 & 0.82 \\
\hline
MSE & 0.07 & 0.07 & 0.08 & 0.02 & 0.02 & 0.01 & \textbf{0.00} & 0.02 & 0.01 & 0.02 & 0.01 \\
\hline
\end{tabular} }
\caption{Accuracy and MSE comparison for Simulation 2. We evaluate models by overall accuracy, accuracy on data from each domain $b_i$ with $i=1,\cdots,5$, and overall MSE for class probability $p$. The best results in each scenario are highlighted in bold.}
\label{tab:acc_s2}
\end{table}

We consider three teacher models $T_1$, $T_2$, $T_3$, each specializing in data drawn from $N(\boldsymbol{\mu}_i,\Sigma)$ for $i=1,2,3$. 
Table \ref{tab:acc_s2} shows the overall classification accuracy, MSE, and domain-specific performance on domain $b_i$, which contains data generated from $N(\boldsymbol{\mu}_i,\Sigma)$, $i=1,\cdots,5$. The results align with those from Simulation 1, demonstrating that the proposed MT-BKD method outperforms all other approaches. In particular, the weighted distillation strategy for combining multiple teachers surpasses the equal-weight strategy in accuracy. 
The advantages of the MT-BKD method are further highlighted in its uncertainty quantification capabilities, as shown by the coverage rates in Fig. \ref{fig:sim_coverage} (b). MT-BKD with the weighting strategy closely matches the assigned credible levels, whereas other methods exhibit significantly larger errors, with discrepancies of up to 10\%.

\begin{figure}[h]
\vspace{-0.1cm}
\centering
\includegraphics[width=0.9\linewidth]{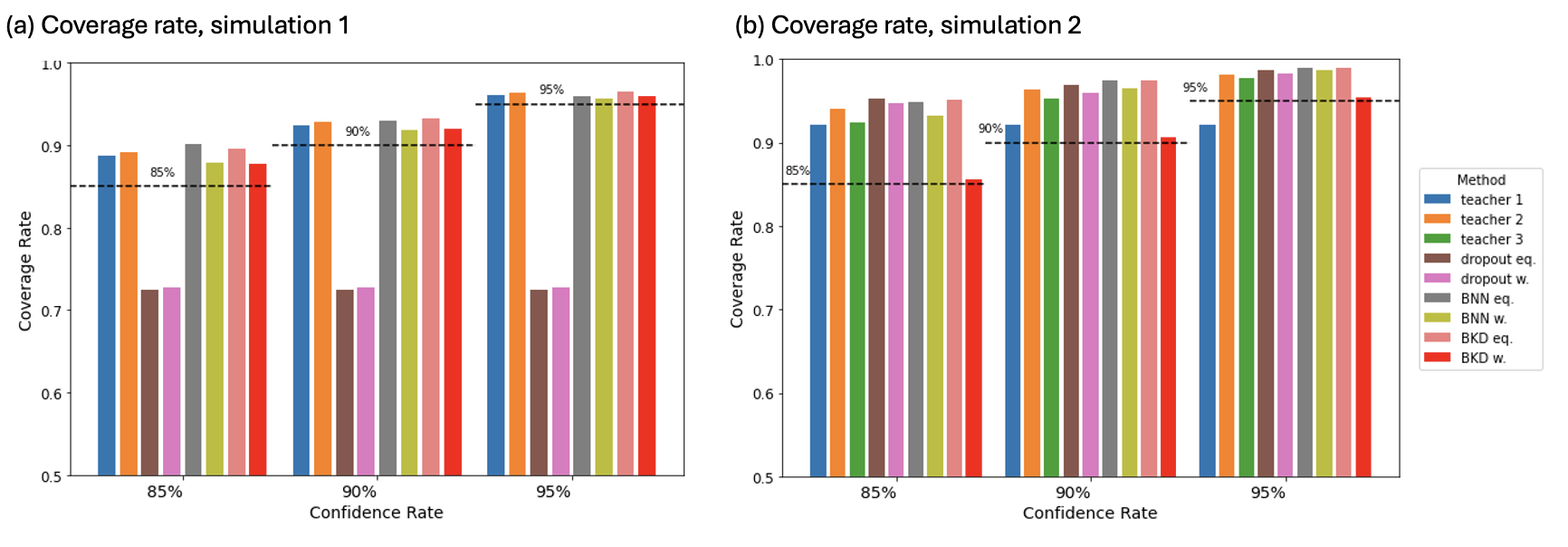}
\vspace{-0.2cm}
\caption{Comparison of coverage rate of (a) simulation 1 and (b) simulation 2 at three standard credible levels (0.85, 0.90, 0.95).  } \label{fig:sim_coverage}
\vspace{-0.2cm}
\end{figure}

\section{Application}\label{sec:real}

We evaluate the performance of MT-BKD on two tasks: eukaryotic protein subcellular localization prediction and image classification. In each task, a compact student model is distilled from teacher models, and we report the student model's uncertainty predictions. Details of the model architecture and training procedures are provided in Appendix \ref{supp:real}.

\subsection{Protein Subcellular location prediction}\label{sec:real_protein}

We apply MT-BKD to predict eukaryotic protein subcellular localization, a key factor in understanding protein function and disease mechanisms. Protein localization provides the necessary context for function, while mislocalization can lead to diseases such as cancer, neurodegenerative, and metabolic disorders \citep{hung2011protein}. Accurate localization prediction is therefore vital for advancing both basic research and therapeutic development.

\begin{figure}[ht!]
\centering
\vspace{2pt}
\includegraphics[width=0.95\linewidth]{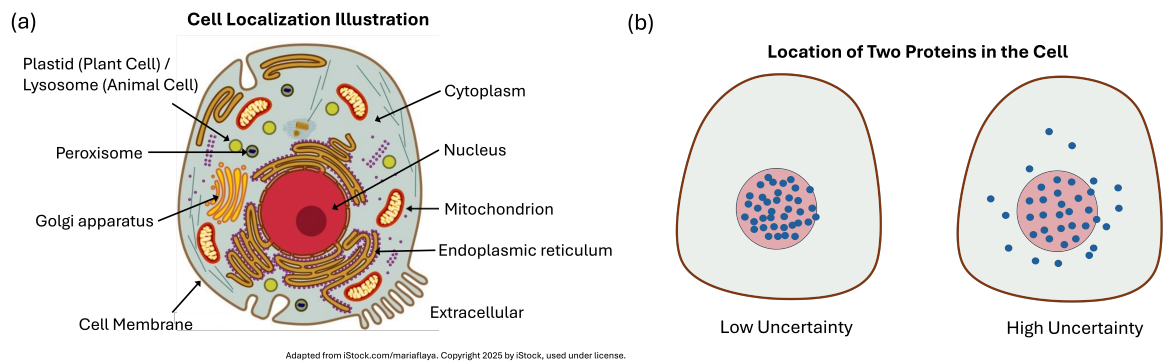}
\caption{Data description. (a) Ten eukaryotic subcellular compartments for the localization task. (b) Protein-specific predictive uncertainty. Low-uncertainty cases support routine annotation, while high-uncertainty cases suggest ambiguous or multi-compartment behavior and merit follow-up.}\label{fig:data_describe}
\end{figure}



\begin{table}
\centering
\scalebox{0.7}{
\begin{tabular}{lcl}
\hline Location & Number of proteins & \multicolumn{1}{c}{ Sublocations } \\
\hline Nucleus & 1243 & Envelope, inner and outer membrane, matrix, lamina, chromosome, \\
& & nucleus speckle \\
Cytoplasm & 955 & Cytoplasm (cytosol and cytoskeleton) \\
Extracellular & 717 & Extracellular \\
Mitochondrion & 515 & Envelope, inner and outer membrane, matrix, intermembrane space \\
Cell membrane & 619 & Apical, apicolateral, basal, basolateral, lateral, cell membrane, cell projection \\
Endoplasmic reticulum & 334 & ER membrane and lumen, microsome, rough ER, smooth ER, \\
& & Sarcoplasmic reticulum \\
Plastid & 247 & Plastid membrane, stroma and thylakoid \\
Golgi apparatus & 142 & Golgi apparatus membrane and lumen \\
Lysosome/Vacuole & 171 & Contractile, lytic and protein storage vacuole, vacuole lumen and membrane,  \\
& & lysosome lumen and membrane \\
Peroxisome & 57 & Peroxisome matrix and membrane \\
\hline
\end{tabular} }
\caption{Number of proteins in each location and sublocations that were grouped together under the same main location.}
\label{tab:data}
\end{table}

We analyze a subset of the DeepLoc2 dataset \citep{thumuluri2022deeploc}, derived from UniProt (release 2021\_03) \citep{uniprot2018uniprot}, which classifies proteins into ten subcellular locations: Cytoplasm, Nucleus, Extracellular, Cell Membrane, Mitochondrion, Plastid, Endoplasmic Reticulum, Lysosome/Vacuole, Golgi Apparatus, and Peroxisome. We focus on sequences from Metazoa, Fungi, or Viridiplantae that range from 40 to 1000 amino acids, excluding proteins assigned to multiple locations. From this filtered set, we randomly select 5,000 sequences. Table \ref{tab:data} summarizes the subcellular locations and corresponding sequence counts. The data is then split into training, validation, and test sets in a 6:2:2 ratio.
We adopt two teacher models with high-capacity protein sequence representations: ESM-2-650M \citep{lin2023evolutionary} (650 million parameters) and ProtT5-XL-half (120 million parameters), alongside a lightweight student model, ESM-2-8M (8 million parameters). 
A classifier is added after the feature extraction layers of each model to adapt their pre-trained embeddings to subcellular location prediction.

\begin{figure}[h]
\centering
\includegraphics[width=0.85\linewidth]{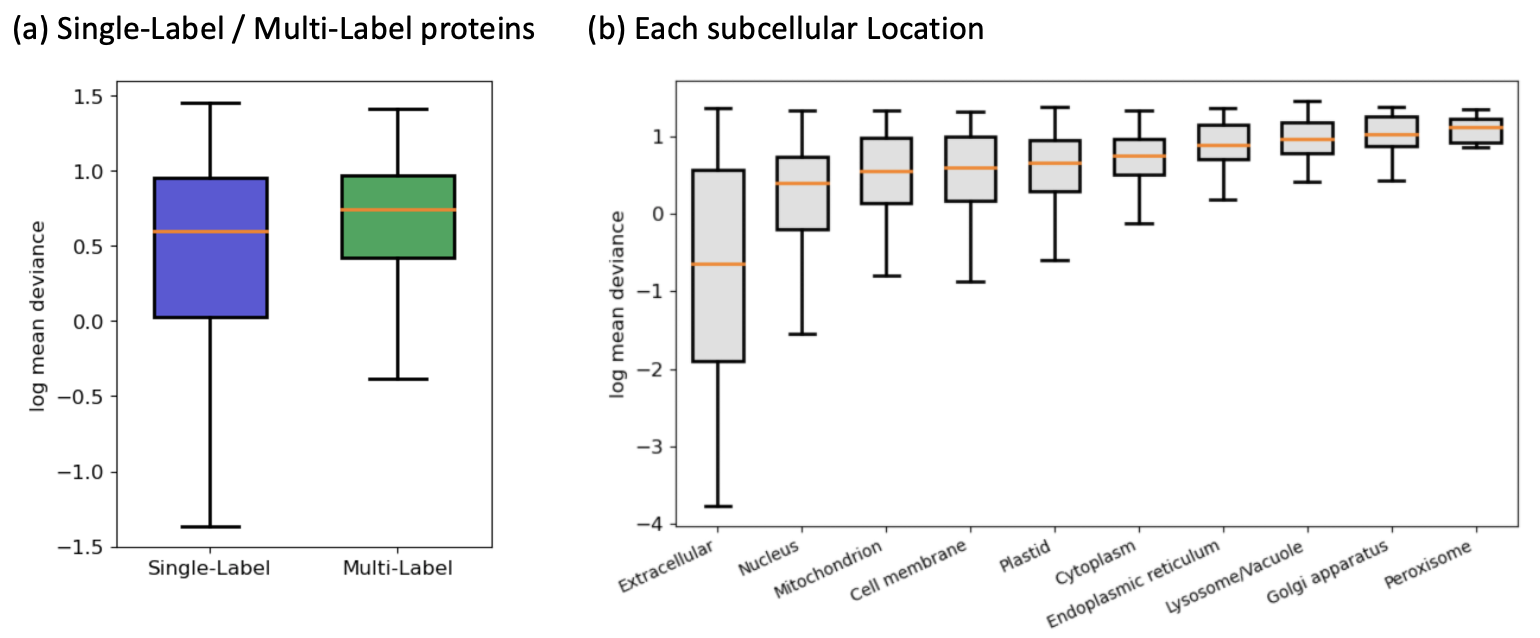}
\vspace{-0.2cm}
\caption{Distribution of log-transformed mean deviance. (a) The first box shows results for the test set, containing only single-label proteins. The second box presents results for a separate dataset of multiple-label proteins. (b) Results for individual protein classes within the single-label test set. } 
\label{fig:protein_box_class}
\vspace{-0.2cm}
\end{figure}

We explore the performance of MT-BKD in uncertainty quantification. 
Fig. \ref{fig:protein_box_class} shows the log-transformed mean deviance distribution. In subfigure (a), the first box shows results from the test set of exclusively single-label proteins, aligning with the training data. The second box displays results from a separate multi-label protein dataset for comparison. As expected, uncertainty is higher for multi-label proteins. This increased uncertainty likely comes from the greater complexity of multiple-label proteins, making them harder to predict. Moreover, since our model was trained only on single-label proteins, reduced performance on multi-label data is unsurprising. The elevated uncertainty appropriately indicates the model’s recognition of its own predictive limitations.
Furthermore, we evaluate the model’s prediction confidence across different classes. In Fig. \ref{fig:protein_box_class} (b), each box corresponds to a specific class in the test set.
MT-BKD exhibits the highest confidence for Extracellular, Nuclear, and Mitochondrial, and the lowest for Golgi apparatus, Lysosomes, and Peroxisomes.
Such observations align with biological evidence that locations like Extracellular, Nuclear, and Mitochondrial are typically easier to identify with high confidence \citep{yogev2011dual, lu2021types, owji2018comprehensive}. This is because of their distinct, well-characterized targeting signals that are consistent across species and highly conserved, i.e., they remain very similar or unchanged across different species over long evolutionary timescales.
In particular, nuclear proteins contain one or more short, positively charged nuclear localization signals (NLS) recognized by importins \citep{dingwall1991nuclear, lu2021types}, while mitochondria-bound proteins typically possess an N-terminal targeting sequence \cite{yogev2011dual}, 
and extracellular proteins are generally marked by a clear N-terminal signal peptide that directs them through the secretory pathway \citep{owji2018comprehensive}.
On the contrary, locations such as the Golgi apparatus, lysosomes, and peroxisomes present challenges due to more diverse and less conserved signals.
Specifically, Golgi proteins often undergo complex modification processes, such as glycosylation, which can vary between species and even cell types \citep{he2024glycosylation}.
Lysosomal targeting relies on tags like mannose-6-phosphate, but not all lysosomal proteins consistently carry this motif \citep{braulke2009sorting}. 
Peroxisomal proteins are similarly difficult to predict due to a varied array of peroxisomal targeting signals (PTS1 and PTS2) with structural nuances \citep{schafer2004functional}.
These biological facts support the biological basis for varying certainty in subcellular localization predictions given by MT-BKD. 

The uncertainty scores achieved by MT-BKD can be leveraged to prioritize experimental validation of borderline predictions, guide model refinement by highlighting where additional data or improved feature engineering is needed, and ultimately enhance our understanding of how subcellular localization signals shape protein distribution.

\subsection{Image Classification}\label{sec:image}

In this study, we utilized the Digit-Five dataset, which contains digit images from multiple domains, each presenting unique characteristics \citep{peng2019moment}. There are 10 classes corresponding to digits ranging from 0 to 9 in each domain. 
Specifically, we focused on three colored variants of the dataset, where each color scheme corresponds to a distinct domain, thereby introducing variations in data distribution and visual representation. These colored variants aim to simulate real-world challenges in domain adaptation by altering pixel intensity patterns, backgrounds, or foreground colors, which can significantly impact model performance.

Each teacher model was trained to specialize in one of the colored datasets, effectively becoming an expert in its respective domain. Using the proposed MT-$\ours$ approach, we trained a single student model to generalize across all three color domains, enabling it to learn from the unique expertise of each teacher. The teacher models use a ResNet architecture with approximately 25.6 million parameters, while the student model employs a convolutional neural network with around 0.1 million parameters, which is 256 times fewer than the teacher models. Additional details about the data and models can be found in Appendix \ref{supp:real}.

\begin{figure}[h]
\vspace{-0.2cm}
\centering
\includegraphics[width=0.6\linewidth]{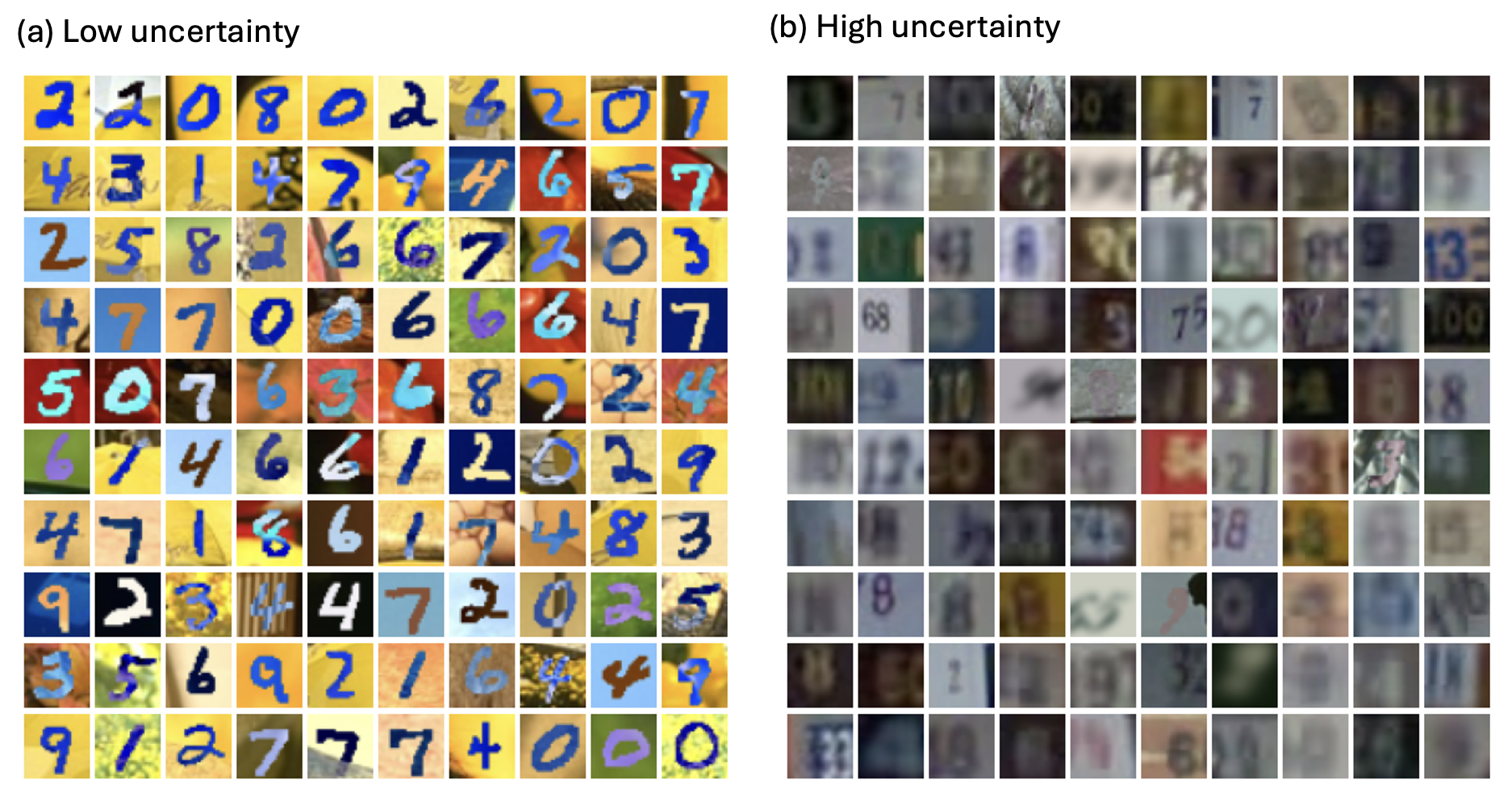}
\vspace{-0.2cm}
\caption{Left panel showcases images with the lowest uncertainty, while the bottom panel showcases images with the highest uncertainty. } \label{fig:digit5}
\end{figure}

In Fig. \ref{fig:digit5}, the images with the lowest and highest uncertainty are shown on the left and right panels, respectively. The images with higher uncertainty are noticeably more difficult to classify. The clear association between high uncertainty and the inherent difficulty in image recognition validates the performance of MT-$\ours$ in uncertainty quantification.

\section{Conclusion}\label{sec:conc}

In this work, we develop a Multiple-teacher Bayesian Knowledge Distillation (MT-$\ours$) framework aimed at compressing the teacher model and providing a suite of Bayesian inference tools for the student model's uncertainty quantification. MT-$\ours$ constructs a teacher-informed prior by integrating external knowledge from multiple teacher models with task-specific training data for the student model’s parameters, from which a posterior is derived. We show that minimizing KD loss is related to estimating the posterior mode in MT-$\ours$, offering a clearer interpretation of how KD operates and suggesting new ways to enhance KD methods. The developed MT-BKD framework contributes to the growing demand for reliable, privacy-aware, and resource-efficient machine learning models in real-world applications.


We evaluate the performance of MT-$\ours$ using both synthetic and real datasets, focusing on classification accuracy and uncertainty quantification ability. The results demonstrate that MT-$\ours$ achieves performance comparable to the teacher model with the advantage of a significantly smaller size.
A key strength of MT-$\ours$ is its ability to precisely quantify prediction uncertainty, as demonstrated by three main findings. First, we confirm that mean deviance is a reliable metric for uncertainty by examining its relationship to the true class probability $p$ and the data’s intrinsic structure (see Fig. \ref{fig:s1_mean_dev}). Second, an analysis of coverage rates highlights the robustness of MT-BKD. Third, in experiments on real datasets, visualizations and distribution analyses of uncertainty estimates across classes further affirm the reliability of MT-BKD’s uncertainty measures.


MT-$\ours$ could be extended to regression and data generation tasks by imposing suitable teacher-informed priors. For regression, one can adopt a Gaussian prior whose mean aligns with the teacher models’ predictions. For score-based generative models, similarly, a Gaussian prior on the score function can be set to have its mean match the teacher-estimated scores.

\putbib[ref]
\newpage
\end{bibunit}
\appendix
\newpage

\begin{bibunit}[apalike]
\input{Supple}
\end{bibunit}

\end{document}

%% file: math_commands.tex
\newcommand{\ours}{\mbox{BKD}}

\renewcommand{\hat}{\widehat}
\newcommand{\bfm}[1]{\ensuremath{\mathbf{#1}}}

\newcommand{\logit}{\text{logit}}

\def \bfy{\mathbf y}

\def \bstheta{\bfsym \theta}

     \def\RR{\mathbb{R}}

\def\bx{\bfm x}

 \def\cD{{\cal  D}}
 \def\cE{{\cal  E}}

 \def\cL{{\cal  L}}

\def\a\cos{\mathrm{arc\cos}}

\newcommand{\bfsym}[1]{\ensuremath{\boldsymbol{#1}}}

\def\bfh{\boldsymbol h} 
\def\boldq{\boldsymbol q}
\def\boldp{\boldsymbol p}

\def\RR{\mathbb R}

\def\bstheta{\bfsym {\theta}}

\def \CE {\text{CE}}

\def\bstheta{\bfsym \theta}

%% file: Supple.tex

\begin{center}
  {\Large\bf APPENDIX}  
\end{center}

\setcounter{page}{1} 
\numberwithin{equation}{section}

\noindent Here is an outline of the Appendix.

\begin{itemize}
    \item Appendix \ref{supp:notation} summarizes the notations used in the main text.
    \item Appendix \ref{supp:add} discusses the single teacher situation.
    \item Appendix \ref{supp:proof} gives the proof of theoretical results.
    \item Appendix \ref{supp:simu} presents supplementary details for simulation experiments.
    \item Appendix \ref{supp:real} presents supplementary details for real data analysis.
\end{itemize}

\newpage
\section{Notation}\label{supp:notation}

In Table \ref{tab: notation}, we summarize the notations in the main text.

{
\begin{table}[H]
\centering
\begin{tabular}{c|l}
\hline
Notation & Interpretation \\ \hline
$\mathcal{F}$ & function classes\\
$\mathbb{P}$ & population distribution \\
$\mathcal{D}, \mathcal{T}$ & dataset \\
$\mathcal{P}$ & distribution class\\
$\bfh$ & function representing the student neural network \\
$\bstheta$ & parameters of the student neural network \\
$\bstheta^*$ & estimate of $\bstheta$ \\
$\hat\bstheta$ & Monte Carlo sample of $\bstheta$ \\
$\bx$ & vector representing $m$-dimensional feature \\
$y$ & class label \\
$\bfy$ & one-hot encoding of class label $y$ \\
$\boldp$ & output of the tea neural network, representing the predicted class probability \\
$\boldq$ & output of the student neural network, representing the predicted class probability \\
$N$ & sample size of the training dataset $\mathcal{D} $ \\
$n$ & sample size of the testing dataset $\mathcal{T} $ \\
$l$ & likelihood function \\
$\nabla$ & gradient operator \\
$\lambda$ & weight parameter \\
$\operatorname{dev}$ & deviance \\
\hline
\end{tabular}
\caption{Notation table}\label{tab: notation}
\end{table}
}

\section{Additioanl Discussion}\label{supp:add}

\subsection{Single Teacher Case}\label{supp:add_single}

In the case where $G=1$, meaning there is only a single teacher model, we present the following theorem, which establishes the connection between our established Bayesian posterior and the original KD framework. 
\begin{theorem}\label{thm:equiv}
    Assume the prior distributions are defined in Equation (\ref{eq:prior_KD_theta}) and (\ref{eq:prior_KD_q}). When there is only one teacher model, represented by $G = 1$, the posterior mode of $\bstheta$ is the minimizer of KD, i.e., $\bstheta_{KD}^*$.
\end{theorem}
\begin{proof}
The negative log-transformed posterior distribution of $\bstheta$ is
\begin{small}
\begin{align}\label{eq:ll_q}
        -l(\bstheta \mid\mathcal{D}, \boldp, \lambda,  \bfh(\cdot;\cdot) ) 
         =&  -\sum_{i=1}^N  \sum_{k=1}^K y_{ik}\log\left(h_k(\bx_i;\bstheta)\right)  \nonumber\\
        &- \lambda \sum_{i=1}^N \sum_{k=1}^K  p_{ik}\log \left(h_k(\bx_i;\bstheta)\right)  + c,\nonumber\\
         =&\  \mathcal{L}^{\mathrm{KD}}( \bfh(\cdot , \bstheta)\mid \mathcal{D}, \boldp,  \lambda) +c,
\end{align} 
\end{small}
where $c$ is a constant. It is easy to verify that the posterior mode of $\bstheta$ is the minimizer of the KD loss in Equation (\ref{eq:opt_KD_loss}).
\end{proof}

In other words, optimizing the KD loss to search for parameters $\bstheta$ in the student model is tantamount to achieving the maximum likelihood estimation of the parameters $\bstheta$.

\section{Proof of Theoretical Results}\label{supp:proof}

\noindent\textbf{Proposition \ref{lemma:proper}.}
    \textit{Consider the probability density function $f(\boldq\mid\boldp_{i}^{(g)})$  as defined in Equation (\ref{eq:prior_KD_q}) with a constant $\lambda > 0$. 
    Assuming that the parameters of the student model lie in a compact space, then $\pi(\bstheta\mid\{\boldp_{i}^{(g)}\}_{i=1,\cdots,N}^{g=1,\cdots,G}) $ is a proper prior.
    }

\begin{proof}
According to Equation (\ref{eq:prior_KD_theta}) and (\ref{eq:prior_KD_q}), for some constant $C>0$, we have 
\begin{align}
    \pi(\bstheta\mid\{\boldp_{i}^{(g)}\}_{i=1,\cdots,N}^{g=1,\cdots,G}) 
    &= C\ \prod_{i=1}^N \sum_{g=1}^G f( \bfh(\mathbf{x}_i; \bstheta)\mid\boldp_{i}^{(g)}) \nonumber \\
    &= C\ \prod_{k=1}^K \sum_{g=1}^G b_{ij} \prod_{i=1}^N \left(\bfh(\mathbf{x}_i; \bstheta) \right)
    ^{\lambda p_{i k}^{(g)}}, 
\end{align}
where $b_{ig} \triangleq \frac{1}{B(\mathbf{1}_K+\lambda \boldp_{i}^{(g)})}>0$ is a constant.
Taking integration, we get
\begin{equation}\label{eq:intg}
    \int  \pi(\bstheta\mid\{\boldp_{i}^{(g)}\}_{i=1,\cdots,N}^{g=1,\cdots,G})  d\bstheta 
    =  C\int \prod_{k=1}^K \sum_{g=1}^G b_{ig} \prod_{i=1}^N \left(\bfh(\mathbf{x}_i; \bstheta) \right)^{\lambda p_{i k}^{(g)}} d\bstheta .
\end{equation}
Since $0 \leq \bfh(\mathbf{x}_i; \bstheta) \leq 1$, we have $0 \leq  \prod_{k=1}^K \sum_{g=1}^G b_{ig} \prod_{i=1}^N\left(\bfh(\mathbf{x}_i; \bstheta) \right)^{\lambda p_{i k}^{(g)}}  \leq 1$ for $1\leq k \leq K$, $1\leq g \leq G$ and $1\leq i\leq N$. 
Given that $\bstheta$ lies in a compact space and $ \prod_{k=1}^K \sum_{g=1}^G b_{ig} \prod_{i=1}^N\left(\bfh(\mathbf{x}_i; \bstheta) \right)^{\lambda p_{i k}^{(g)}} $ is bounded, we have 
\begin{equation}
    \int \prod_{k=1}^K \sum_{g=1}^G b_{ig} \prod_{i=1}^N \left(\bfh(\mathbf{x}_i; \bstheta) \right)^{\lambda p_{i k}^{(g)}} d\bstheta < \infty.
\end{equation}
Thus, we can easily verify that $\pi(\bstheta\mid\{\boldp_{i}^{(g)}\}_{i=1,\cdots,N}^{g=1,\cdots,G})$ is a proper prior.
\end{proof}


\noindent\textbf{Proposition \ref{thm:converge}.}
    \textit{Consider the probability density function $f(\boldq\mid\boldp_{i}^{(g)})$ as defined in Equation (\ref{eq:prior_KD_q}), as $\lambda \to \infty$, we have
    \begin{equation}
        f(\boldq \mid\boldp_{i}^{(g)}) \longrightarrow \delta(\boldq-\boldp_{i}^{(g)}) ,
    \end{equation}
    where $\delta(\cdot)$ is the multivariate Dirac delta function.}

\begin{proof}
For the sake of brevity, we consider the general formula $f(\boldq\mid\boldp,\lambda)$, where $\boldq=(q_1,q_2,\cdots,p_K)^T$ and $\boldp=(p_1,p_2,\cdots,p_K)^T$. WLOG, we consider the situation where $p_k \neq 0$ for $k \in 1,\cdots,K$. We have 
\begin{equation}
    f(\boldq\mid\boldp,\lambda) 
    = \frac{1}{B(\mathbf{1}_K+\lambda \boldp)}\Pi_{k=1}^K (q_{k})^{\lambda p_{k}} 
    = \frac{\Gamma(\lambda + K)}{\prod_{k=1}^K\left(\Gamma(\lambda p_k+1)\right) }\Pi_{k=1}^K (q_{k})^{\lambda p_{k}}.
\end{equation}

\noindent We start with the situation when $\lambda p_k, k=1, \cdots, K$, are all integers. We have
\begin{equation}\label{suppeq:pdf}
    f(\boldq\mid\boldp,\lambda)=\frac{(\lambda+K-1) !}{\prod_{k=1}^K\left(\lambda p_k\right) !} \cdot \prod_{k=1}^K\left(q_k\right)^{\lambda p_k}.
\end{equation}
According to the Stirling formula that
$ n ! \sim \sqrt{2 \pi n}\left(\frac{n}{e}\right)^n $,
we have 
\begin{equation}
 \begin{aligned}
    f(\boldq\mid\boldp,\lambda) &\sim  \frac{\lambda^{K-1} \cdot \sqrt{\lambda}\left(\frac{\lambda}{e}\right)^\lambda}{\prod_{k=1}^K \sqrt{\lambda p_k}\left(\frac{\lambda p_k}{e}\right)^{\lambda p_k}} \cdot \prod_{k=1}^K\left(q_k\right)^{\lambda p_k} \\ 
    &\sim \lambda^{\frac{K-1}{2}} \cdot \prod_{k=1}^K\left(\frac{q_k}{p_k}\right)^{\lambda p_k}.
\end{aligned}   
\end{equation}
Write $g(\boldq)=\prod_{k=1}^K\left(\frac{q_k}{p_k}\right)^{p_k}$, we now get the maximization of $g(\boldq)$ w.r.t. $\boldq$. Since we have the constraint that $\sum_{k=1}^K q_k=1$, we use the Lagrange multiplier to find the maximum of $g(\boldq)$. 
The Lagrangian function is defined as 
\begin{equation}
    \mathcal{L}\left(q_1, \ldots, q_K, c\right) \equiv \log g+c\left(1-\sum_{k=1}^K q_k\right) ,
\end{equation}
where $\log g=\sum_{k=1}^K p_k\left(\log q_k-\log p_k\right)$.

\noindent To solve 
\begin{equation}
    \nabla_{q_1, \cdots, q_K, c} g\left(q_1, \cdots, q_K, c\right)=0,
\end{equation}
we have
\begin{equation}
    \left\{\begin{array}{l}\frac{p_k}{q_k}-c=0, \qquad \text{for } k=1, \ldots, K \\ \sum_{k=1}^K q_k=1\end{array}\right.
\end{equation}
since $\sum_{k=1}^K p_k=1$, we can easily get that the solution is
\begin{equation}\label{suppeq:sol}
\begin{aligned}
    & \left\{\begin{array}{l}
    q_k=p_k, \qquad \text{for } k=1, \cdots, K \\
    c=1.
    \end{array}\right. 
\end{aligned}
\end{equation}
That is, we have $g_{\text{max}}(\boldq)=1$ when Equation (\ref{suppeq:sol}) holds, and this leads to $f(\boldq\mid\boldp,\lambda) \sim \lambda^{\frac{K-1}{2}} \cdot 1^{\lambda} \stackrel{\lambda \rightarrow \infty}{\longrightarrow} \infty$.

\noindent For $\forall \boldq \neq \boldp$, we have $g(\boldq) < 1$, which leads to  $f(\boldq\mid\boldp,\lambda) \sim \lambda^{\frac{K-1}{2}} \cdot \left(g(\boldq)\right)^{\lambda} \stackrel{\lambda \rightarrow \infty}{\longrightarrow} 0$.
Thus, as $\lambda \rightarrow \infty$, we have
\begin{equation}
    \begin{aligned}
        & \lim _{\lambda \rightarrow \infty} \int f(\boldq\mid\boldp,\lambda) \mathrm{d} \boldq = 1 , \\
        & f(\boldq\mid\boldp,\lambda) = 0 \text{ if } \boldq-\boldp\neq 0 .
    \end{aligned}
\end{equation}
That is, $f(\boldq \mid\boldp,\lambda) \longrightarrow \delta(\boldq-\boldp)$ as $\lambda \rightarrow \infty$.

If $\exists \ k$ s.t. $\lambda p_k$ is not an integer.
According to the Stirling formula for the gamma function that $\Gamma(z) \sim \sqrt{\frac{2 \pi}{z}}\left(\frac{z}{e}\right)^z$, following the same analysis, we can also get $f(\boldq \mid\boldp,\lambda) \longrightarrow \delta(\boldq-\boldp)$ as $\lambda \rightarrow \infty$.

\end{proof}

\noindent\textbf{Theorem \ref{thm:equiv_multi}.}
    Given the prior distribution defined in Equation (\ref{eq:prior_KD_theta}) and (\ref{eq:prior_KD_q}), the traditional multi-teacher KD strategy in Equation (\ref{eq:opt_multi_KD_loss}) is equivalent to deriving the posterior distributions for individual teacher models separately and combining them using logarithmic opinion pooling.

\begin{proof}
Based on a single $g^{\text{th}}$ teacher model, we have the prior distribution $\pi_{(g)}(\bstheta; \{\boldp_i^{(g)}\}_{i=1,\cdots,N}) \propto \prod_{i=1}^N f(\bfh(\bx_i; \bstheta); \boldp_i^{(g)})$, where $g = 1, \dots, G$. The posterior density is given by
\begin{small}
\begin{align}
    \pi_{(g)}(\bstheta \mid\mathcal{D}, \boldp^{(g)}, \lambda,  \bfh(\cdot,\cdot) ) 
    & \propto \prod_{i=1}^N \prod_{k=1}^K h_k\left(\bx_i; \bstheta\right)^{y_{i k}} \cdot
    \prod_{i=1}^N \frac{1}{B\left(\mathbf{1}_k + \lambda\boldp_i^{(q)}\right)} \prod_{k=1}^K\left(h_k\left(\bx_i; \bstheta\right)\right)^{\lambda p_{i k}{ }^{(g)}}.
\end{align}
\end{small}
Combining the posterior distributions based on each teacher model through a logarithmic pooling strategy, we have the combined posterior
\begin{small}
    \begin{align}
        \pi_{\text{pooled}}(\bstheta \mid\mathcal{D}, \boldp, \lambda,  \bfh(\cdot,\cdot)) \propto \Pi_{g=1}^G \pi_{(g)}(\bstheta \mid\mathcal{D}, \boldp^{(g)}, \lambda,  \bfh(\cdot,\cdot))^{w^{(g)}},
    \end{align}
\end{small}
where $w^{(g)}$ is the weight for the $g^{\text{th}}$ teacher model. Thus, the negative log-transformed posterior distribution of $\bstheta$ is
\begin{small}
\begin{align}
        -l(\bstheta \mid\mathcal{D}, \boldp, \lambda,  \bfh(\cdot,\cdot) ) 
         =&  -\sum_{i=1}^N  \sum_{k=1}^K y_{ik}\log\left(h_k(\bx_i;\bstheta)\right)  \nonumber\\
        &- \lambda \sum_{g=1}^G w^{(g)} \sum_{i=1}^N \sum_{k=1}^K  p_{ik}^{(g)}\log \left(h_k(\bx_i;\bstheta)\right)  + c,\nonumber\\
         =&\  \mathcal{L}^{\mathrm{Multi-KD}}( \bfh(\cdot , \bstheta)\mid \mathcal{D}, \boldp,  \lambda) +c,
\end{align} 
\end{small}
where $c$ is a constant. That is, finding the mode of the combined posterior, obtained by merging the posterior distributions from each teacher model using a logarithmic pooling strategy, is equivalent to minimizing the multi-teacher KD loss described in Equation (\ref{eq:opt_multi_KD_loss}).
\end{proof}

\noindent\textbf{Theorem \ref{thm:asymp}.}
The estimate of $\bstheta$ by maximizing Equation (\ref{eq:opt_multi_KD_loss}) yields a consistent estimate of $\bfh(\bx;\bstheta)$ with the optimal convergence rate upper bounded by  
\begin{equation}
    C_2(n) \wedge   G\cdot V \left( \left\{\sum_{g=1}^G \left[\frac{\mu^{(g)}(\bx)}{\mu(\bx) \varepsilon^{(g)} } \right]^{\frac{1}{2}} \right\}^{-1}\right)  ,
\end{equation}
where $C_2(n) \asymp C_1(n)$. 

(A1) For any $\bx \in \Omega_\bx$,  $\mu(\bx)>0$, $\max_{ 1\leq g \leq G} \mu^{(g)}(\bx) >0$. 

(A2) The minimizer $\hat{\bstheta}_n$ in Equation (\ref{eq:loss}) yield a consistent estimate of $\bfh(\bx;\bstheta)$, i.e., the error $\int_{\Omega_\bx} [\bfh(\bx;\hat{\bstheta}_n) - \bfh(\bx;\bstheta) ]^2 \mu(\bx) d\bx  = O_p(C_1(n)) $ and $\lim_{n \to \infty } C_1(n) = 0$.  

(A3) There exist a  positive constant $\varepsilon^{(g)}$ such that $\int_{\Omega_\bx} [M_t^{(g)}(\bx) - \bfh(\bx;\bstheta) ]^2 \mu^{(g)}(\bx) d\bx \leq \varepsilon^{(g)} $ for $g = 1,\ldots,G$.

\begin{proof}
    Let $\hat{\bstheta}_n$ be the maximize of Equation (\ref{eq:opt_multi_KD_loss}).
We define
\begin{eqnarray*}
    \cE_n^{\varphi}(\bstheta)=\frac{1}{N} \sum_{i=1}^N \varphi\left(\bfy_i^T \logit [\bfh(\bx_i;\bstheta)]\right)
\end{eqnarray*}
Notice that $\cE_n^{\varphi}(\bstheta) = \cL(\bstheta|\cD)$ and 
\begin{eqnarray*}
    \cE^{\varphi}(\bstheta) &=& E_\bx(KL[\bfh(\bx;\bstheta_0) ,\bfh(\bx;\bstheta) ])\\
    &=&  E_{\bx,\bfy}(CE[\bfy ,\bfh(\bx;\bstheta) ]) - E_{\bx}(CE[\bfh(\bx_0;\bstheta),\bfh(\bx;\bstheta_0)])
\end{eqnarray*}

For $g$th teacher
    \begin{eqnarray*}
       CE[\bfh(\bx;\bstheta_0),\bfh(\bx;\bstheta)]) &\leq& CE[M_t^{(g)}(\bx),\bfh(\bx;\bstheta)] + \|M_t^{(g)}(\bx) - \bfh(\bx;\bstheta_0) \|_2\|\log \bfh(\bx;\bstheta)  \|_2  \\
        &\leq& CE[M_t^{(g)}(\bx),\bfh(\bx;\bstheta)] + D_1  \|M_t^{(g)}(\bx) - \bfh(\bx;\bstheta_0) \|_2
    \end{eqnarray*}

Without loss of generality, we consider the function 
\begin{eqnarray*}
    J(\bstheta;\bx) = \log \left[\sum_{g=1}^G \tilde{w}^{(g)}(\bx) \exp(\lambda CE[M^{(g)}_t(\bx),\bfh(\bx;\theta) ]  )\right],
\end{eqnarray*}
where $\tilde{w}^{(g)}(\bx)= w^{(g)}(\bx)/{B\left(\mathbf{1}_k+\lambda M^{(g)}_t(\bx)\right)}$.

suppose we set $\tilde{w}^{(g)}(\bx) \propto 1/\exp(\lambda CE[M^{(g)}_t(\bx),\bfh(\bx;\theta_0) ]  )$.Let $W(\bx) = \sum_{g=1}^G \tilde{w}^{(g)}$. 
Since  $\sum_{g=1}^G w^{(g)} = 1$ and $B\left(\mathbf{1}_k+\lambda M^{(g)}_t(\bx)\right)$ is bounded, we may assume $\frac{1}{W(\bx)} < D_2$ for a constant $D_2>0$.
By Jensen's inequality, we have
\begin{eqnarray*}
    J(\bstheta_0;\bx) &=& -\log\left[\frac{W(\bx)}{G} \sum_{g=1}^G \exp(-\lambda CE[M^{(g)}_t(\bx),\bfh(\bx;\theta_0) ]  ) \right]\\
    &\leq& \lambda \frac{W(\bx)}{G} \sum_{g=1}^G CE[M^{(g)}_t(\bx),\bfh(\bx;\theta_0) ]
\end{eqnarray*}

By Hölder's inequality, we have
\begin{eqnarray*}
   &&  \sum_{g=1}^G \tilde{w}^{(g)}(\bx) \exp(CE[\bfh(\bx;\bstheta_0),\bfh(\bx;\bstheta)]) \\
    &=&  \sum_{g=1}^G \tilde{w}^{(g)}(\bx) \exp(CE[M^{(g)}_t(\bx),\bfh(\bx;\bstheta)]) \exp( [M^{(g)}_t(\bx)-\bfh(\bx;\bstheta_0)]^T\log \bfh(\bx;\bstheta) ) \\ 
    &\leq & \left[\sum_{g=1}^G \tilde{w}^{(g)}(\bx) \exp(\lambda CE[M^{(g)}_t(\bx),\bfh(\bx;\bstheta)]) \right]^{\frac{1}{\lambda}} \left[\sum_{g=1}^G \tilde{w}^{(g)}(\bx) \exp( \lambda[M^{(g)}_t(\bx)-\bfh(\bx;\bstheta_0)]^T\log \bfh(\bx;\bstheta) ) \right]^{\frac{1}{\lambda}}  \\
    &\leq & \exp{J(\bstheta;\bx)} \left[\sum_{g=1}^G \tilde{w}^{(g)}(\bx) \exp( [M^{(g)}_t(\bx)-\bfh(\bx;\bstheta_0)]^T\log \bfh(\bx;\bstheta) ) \right]  \\
    &\leq& \exp{J(\bstheta;\bx)} \exp(T_1(\bx,\bstheta)) \\
 \end{eqnarray*}

\begin{eqnarray*}
     \cE^{\varphi}(\bstheta) &\leq&  E_\bx \left[\frac{1}{W(\bx)}\sum_{g=1}^G \tilde{w}^{(g)}(\bx) CE[M^{(g)}_t(\bx),\bfh(\bx;\theta) ] +   \sum_{g=1}^G \frac{\tilde{w}^{(g)}(\bx)}{W(\bx)} \|M_t^{(g)}(\bx) - \bfh(\bx;\bstheta_0) \|_2  \right] \\
     &\leq& D_2 E_\bx[J(\bstheta;\bx)]+ \sum_{g=1}^G \int\frac{\tilde{w}^{(g)}(\bx)}{W(\bx)} \|M_t^{(g)}(\bx) - \bfh(\bx;\bstheta_0) \|_2 \mu(\bx)d\bx  \\
     &\leq& D_2 E_\bx[J(\bstheta;\bx)]+    \sum_{g=1}^G 
 \sqrt{\int  \frac{[\tilde{w}^{(g)}(\bx)\mu(\bx)/W(\bx)]^2}{\mu^{(g)}(\bx)}  d\bx  \int  \|M_t^{(g)}(\bx) - \bfh(\bx;\bstheta_0) \|_2^2  \mu^{(g)}(\bx)d\bx  } \\
 &\leq&  D_2 E_\bx[J(\bstheta;\bx)]+   \sum_{g=1}^G 
 \sqrt{\varepsilon^{(g)} \int  \frac{[\tilde{w}^{(g)}(\bx)\mu(\bx)/W(\bx)]^2}{\mu^{(g)}(\bx)}  d\bx    },
\end{eqnarray*}

We have 
\begin{eqnarray*}
   && \log \left[\sum_{g=1}^G \tilde{w}^{(g)}(\bx) \exp( [M^{(g)}_t(\bx)-\bfh(\bx;\bstheta_0)]^T\log \bfh(\bx;\bstheta) ) \right] \\
   &\leq&  \log \left[\sum_{g=1}^G \tilde{w}^{(g)}(\bx) \exp( \|M^{(g)}_t(\bx)-\bfh(\bx;\bstheta_0)\|_2^2/2 + \|\log \bfh(\bx;\bstheta) \|_2^2/2 )\right] \\
   &\leq& D_1  \log \left[\sum_{g=1}^G \tilde{w}^{(g)}(\bx) \exp( \|M^{(g)}_t(\bx)-\bfh(\bx;\bstheta_0)\|_2^2/2 )\right] 
\end{eqnarray*}

By Jensen's inequality, we have
\begin{eqnarray*}
J(\bstheta;\bx)
   &\geq& \sum_{g=1}^G \tilde{w}^{(g)}(\bx) CE[M^{(g)}_t(\bx),\bfh(\bx;\theta) ].
\end{eqnarray*}

Let $W(\bx) = \sum_{g=1}^G \tilde{w}^{(g)}$. 
Since  $\sum_{g=1}^G w^{(g)} = 1$ and $B\left(\mathbf{1}_k+\lambda M^{(g)}_t(\bx)\right)$ is bounded, we may assume $\frac{1}{W(\bx)} < D_2$ for a constant $D_2>0$.
We have

\begin{eqnarray*}
     \cE^{\varphi}(\bstheta) &\leq&  E_\bx \left[\frac{1}{W(\bx)}\sum_{g=1}^G \tilde{w}^{(g)}(\bx) CE[M^{(g)}_t(\bx),\bfh(\bx;\theta) ] +   \sum_{g=1}^G \frac{\tilde{w}^{(g)}(\bx)}{W(\bx)} \|M_t^{(g)}(\bx) - \bfh(\bx;\bstheta_0) \|_2  \right] \\
     &\leq& D_2 E_\bx[J(\bstheta;\bx)]+ \sum_{g=1}^G \int\frac{\tilde{w}^{(g)}(\bx)}{W(\bx)} \|M_t^{(g)}(\bx) - \bfh(\bx;\bstheta_0) \|_2 \mu(\bx)d\bx  \\
     &\leq& D_2 E_\bx[J(\bstheta;\bx)]+    \sum_{g=1}^G 
 \sqrt{\int  \frac{[\tilde{w}^{(g)}(\bx)\mu(\bx)/W(\bx)]^2}{\mu^{(g)}(\bx)}  d\bx  \int  \|M_t^{(g)}(\bx) - \bfh(\bx;\bstheta_0) \|_2^2  \mu^{(g)}(\bx)d\bx  } \\
 &\leq&  D_2 E_\bx[J(\bstheta;\bx)]+   \sum_{g=1}^G 
 \sqrt{\varepsilon^{(g)} \int  \frac{[\tilde{w}^{(g)}(\bx)\mu(\bx)/W(\bx)]^2}{\mu^{(g)}(\bx)}  d\bx    },
\end{eqnarray*}
Notice that if $\bstheta = \bstheta_0$, we have $\cE^{\varphi}(\bstheta_0) \leq  D_2 E_\bx[J(\bstheta_0;\bx)]$.

Suppose we set $\tilde{w}^{(g)}(\bx) \propto \sqrt{\mu^{(g)}(\bx)/\varepsilon^{(g)}}$, we have
\begin{eqnarray*}
    \cE^{\varphi}(\bstheta) &\leq&  D_2 E_\bx[J(\bstheta;\bx)]+  G\cdot V(   \frac{ \sqrt{\mu(\bx)} }{ \sum_{g=1}^G \sqrt{\mu^{(g)}(\bx) /\varepsilon^{(g)}}} )
\end{eqnarray*}

Furthermore, we have 
\begin{eqnarray*}
   E_\bx[J(\bstheta_0;\bx) ]  &\leq& \log \sum_{g=1}^G E_\bx \left[ \tilde{w}^{(g)}(\bx) \exp(\lambda CE[M^{(g)}_t(\bx),\bfh(\bx;\theta_0) ]  )\right],
\end{eqnarray*}

For convenience, we denote  $\frac{1}{N} \sum_{i=1}^N \log \left[\sum_{g=1}^G w_i^{(g)} \frac{1}{B\left(\mathbf{1}_k+\lambda \boldp_i^{(g)}\right)} \prod_{k=1}^K\left(h_k\left(\bx_i; \bstheta\right)\right)^{\lambda p_{ik}^{(g)}}\right]$ by  $\hat{E}_\bx[J(\bstheta;\bx)|\cD]$. Let $T_1 = V(   \frac{ \sqrt{\mu(\bx)} }{ \sum_{g=1}^G \sqrt{\mu^{(g)}(\bx) /\varepsilon^{(g)}}} )$

Hence, we have
    \begin{eqnarray*}
        (1+\lambda)\cE^{\varphi}(\hat{\bstheta}_n,\bstheta_0)  &\leq&  \cE^{\varphi}(\hat{\bstheta}_n,\bstheta_0) + \lambda  E_\bx[J(\hat{\bstheta}_n;\bx)]+\lambda T_1 - \lambda \cE^{\varphi}(\bstheta_0)
        \\
&=& \cE^{\varphi}(\hat{\bstheta}_n,\bstheta_0) - \cE_n^{\varphi}(\hat{\bstheta}_n,\bstheta_0) + \cL^{Multi-KD}(\hat{\bstheta}_n)  - \cL(\bstheta_0|\cD) \\
&& + \lambda E_\bx[J(\bstheta;\bx)]-\lambda \hat{E}_\bx[J(\bstheta;\bx)|\cD] +\lambda T_1 - \lambda \cE^{\varphi}(\bstheta_0)  \\
&\leq&  \left|\cE^{\varphi}(\hat{\bstheta}_n,\bstheta_0) - \cE_n^{\varphi}(\hat{\bstheta}_n,\bstheta_0)\right|  + \cL^{Multi-KD}(\bstheta_0)  - \cL(\bstheta_0|\cD) \\
&& + \lambda \left|E_\bx[J(\bstheta_0;\bx)]- \hat{E}_\bx[J(\bstheta;\bx)|\cD]\right| +\lambda T_1  - \lambda \cE^{\varphi}(\bstheta_0)  \\
&\leq&  \left|\cE^{\varphi}(\hat{\bstheta}_n,\bstheta_0) - \cE_n^{\varphi}(\hat{\bstheta}_n,\bstheta_0)\right|  + \lambda \hat{E}_\bx[J(\bstheta_0;\bx)|\cD]  - \lambda E_\bx[J(\bstheta_0;\bx)]  \\
&& + \lambda \left|E_\bx[J(\bstheta_0;\bx)]- \hat{E}_\bx[J(\bstheta;\bx)|\cD]\right| +\lambda T_1 + \lambda E_\bx[J(\bstheta_0;\bx)] - \lambda \cE^{\varphi}(\bstheta_0) 
    \end{eqnarray*}

    
\end{proof}

\section{Simulation}\label{supp:simu}

\subsection{Details on Model Structure}

\noindent \textbf{Simulation 1:}
The teacher models employ a Multilayer Perceptron (MLP) architecture consisting of five hidden layers. These layers have 7, 10, 12, 10, and 5 nodes respectively. The model uses the ReLU activation function. Different dropout rates are applied for two teacher models.  
The student model employs an MLP architecture consisting of one hidden layer with 5 nodes. 

\noindent \textbf{Simulation 2:}
The teacher models employ the MLP architecture consisting of five hidden layers with 7, 15, 12, 10, and 7 nodes respectively. The model uses the ReLU activation function and incorporates a dropout rate of 0.2. 
The student model utilizes an MLP architecture with a single hidden layer comprising 10 nodes, activated by the ReLU function.

\begin{table}[h]
\caption{Model Structure. The number of parameters for each respective model is indicated in parentheses.}
    \centering
    \scalebox{0.8}{
    \begin{tabular}{c|ccc|c} 
    \hline
    & \multicolumn{3}{c|}{ Teacher models } & Student Model \\
     & $M_T^{(1)}$ & $M_T^{(2)}$ & $M_T^{(3)}$ &  $M_s$ \\ \hline 
    Simulation 1 & MLP-1 (430) & MLP-1 (430) & - & MLP-S1 (27) \\
    Simulation 2 & MLP-2 (600) & MLP-2 (600) & MLP-2 (600) & MLP-S2 (115)  \\
    Protein & ESM-2-650M (650M) & ProtT5-XL-half (120M) & - & ESM-2-8M (8M) \\
    Digit-5 & CNN-L (2M) & Resnet-50 (26M) & Resnet-50 (26M) & CNN-S (0.1M) \\
    \hline
    \end{tabular}}
    \label{tab:model}
    \vspace{-0.1cm}
\end{table}

\subsection{Implementation Details of BNN and Dropout}\label{supp:imple_BNN_Drop}

To implement the method integrating the original KD and BNN, we set the BNN model as the student. We train this model using variational inference with details described in \citet{shridhar2019comprehensive}, and incorporate the KD loss to distill knowledge from the teacher model. An important hyperparameter when training the BNN model is the weight of the CE loss. We perform a grid search over the values $\{0.01, 0.05, 0.1, 0.2\}$ to find the parameter that yields the highest accuracy.
For the CIFAR-10 and CIFAR-100 datasets, due to computational constraints, we add the Bayesian layers only after the feature extraction blocks.

To implement the method integrating the original KD and Monte Carlo dropout, we enable dropout during the testing/inference stage to obtain multiple stochastic predictions.

\section{Real Data Analysis}\label{supp:real}

\subsection{Details on Protein Subcellular Location Prediction}\label{supp:protein}

\textbf{ESM-2 model:} The ESM-2 model (Evolutionary Scale Modeling v2) \citep{lin2023evolutionary}, developed by Meta AI, is a state-of-the-art protein language model trained on large-scale protein sequence datasets, including UniRef50 \citep{suzek2015uniref}. Using a Transformer architecture, it captures intricate sequence patterns and relationships, excelling in tasks such as structure prediction, mutation effect prediction, and protein design. We use the ESM-2-650M variant, which contains 650 million parameters and offers improved performance and scalability over its predecessor, as teacher model 1. The ESM-2-8M variant, which contains 8 million parameters, is used as the student model.

\textbf{ProtT5 model:} The ProtT5 model is a protein language model based on the T5 architecture and trained on UniRef50. It captures both local and global sequence features, excelling in tasks such as structure prediction, function annotation, and sequence classification. We use the ProtT5-XL-half model as teacher model 2, which is a half-precision encoder-only variant optimized for efficient and accurate protein sequence embeddings.

During training, the feature extraction layers of these teacher models are frozen, and only the classifier layers are fine-tuned using our dataset.

\subsection{Details on Digit-5}\label{supp:digit5}

We focused on three colored variants of Digit-5: \textbf{SVHN}, \textbf{MNIST-M}, and \textbf{SynthDigits}. 
\textbf{SVHN} (Street View House Numbers) is a dataset of digit images obtained from real-world house number signs in Google Street View \citep{netzer2011reading}. It includes a training set of 73,257 examples and a test set of 26,032 examples. Each example is a 32 × 32 colored image associated with a label of 10 classes, representing the digits 0 through 9.
\textbf{MNIST-M} is a modified version of the original MNIST dataset, where grayscale digit images are superimposed on randomly chosen colored background patches \citep{peng2019moment}. The dataset retains the original MNIST structure, with a training set of 55,000 examples and a test set of 10,000 examples. Each example is a 28 × 28 colored image associated with a label of 10 classes.
\textbf{SynthDigits} is a synthetic dataset consisting of 25,000 training images and 10,000 test images, where each example is a 32 × 32 colored image of a digit from 0 to 9 \citep{ganin2015unsupervised}. The images are generated with diverse fonts, colors, and background patterns, providing a visually rich and challenging dataset for digit recognition tasks.

\noindent \textbf{Model Structure:}
The teacher model employs the 50-layer ResNet architecture as described by \citet{he2016deep}. 
The ResNet50 model has 25.6 million parameters.
The student model employs a CNN architecture with two convolutional layers and a subsequent fully connected layer. The first convolutional layer has 16 output channels, batch normalization, ReLU activation, and max pooling. The second convolutional layer follows a similar pattern but with 8 output channels. The convolutional layer is then followed by fully connected layers, consisting of a hidden layer with 60 nodes and an output layer with 10 nodes. The student model has approximately 0.1 million parameters, which is 256 times fewer than the teacher models.

\putbib[ref]
